\documentclass{article}
\newif\ifdoubleblind
\doubleblindfalse

\newcommand\figdir{figures/} 
\usepackage{iclr2027_conference,times}
\usepackage[T1]{fontenc}
\usepackage[utf8]{inputenc}
\usepackage{amsmath,amsthm,amssymb,bm}
\usepackage{xcolor,graphicx,tikz,float,booktabs,wrapfig}
\usepackage{needspace}
\usetikzlibrary{arrows.meta}
\usepackage[framemethod=tikz]{mdframed}
\usepackage{hyperref,url}
\hypersetup{hidelinks}

\definecolor{cutblue}{RGB}{25,55,140}
\newif\ifcutmode\cutmodetrue

\newcommand{\psd}[1]{\mathbb{S}^{#1}_{+}}
\newcommand{\Cmat}[1]{C_{#1}}
\newcommand{\Mcap}{{M_\cap^*}}
\newcommand{\Rcap}{{R_\cap^*}}
\newcommand{\Wcup}{{W_\cup^*}}
\newcommand{\srcCh}[1]{P_{#1|Y}}
\newcommand{\PYQ}[1]{P_{Y|Q=#1}}
\newcommand{\preceqY}{\sqsubseteq_Y}
\newcommand{\leL}{\preceq}
\newcommand{\geL}{\succeq}

\newtheorem{thm}{Theorem}
\newtheorem{lem}{Lemma}
\newtheorem{prop}{Proposition}
\newcommand\CondGauss{\mathsf{GC}}
\newmdenv[backgroundcolor=black!5,linecolor=black!25,linewidth=0.5pt,
roundcorner=2pt,innertopmargin=6pt,innerbottommargin=6pt,
innerleftmargin=8pt,innerrightmargin=8pt,skipabove=6pt,skipbelow=6pt]{openproblem}
\allowdisplaybreaks

\newcommand{\eye}{I}

\newcommand{\where}{\text{\;where}\;}

\newcommand{\match}[1]{\ensuremath{\bm{#1}}}

\newcommand{\thmBlackwellStatement}[1]{%
Let $A,B$ be sources with a common target $Y\in\mathbb{R}^{d_Y}$ of
arbitrary distribution, whose posteriors are Gaussian with constant
covariance: $P_{Y|A=a}=\mathcal N(m_A(a),\Sigma_{Y|A})$ and
$P_{Y|B=b}=\mathcal N(m_B(b),\Sigma_{Y|B})$ almost surely, where $m_A,m_B$
are measurable and $\Sigma_{Y|A},\Sigma_{Y|B}\succeq0$. Then,
\begin{equation}
\srcCh{A}\preceqY\srcCh{B}\quad\iff\quad \Sigma_{Y|A}\geL\Sigma_{Y|B}. #1
\end{equation}%
}

\newcommand{\thmRedundancyStatement}[1]{%
Suppose that each source--target pair $(X_i,Y)$ is jointly Gaussian, with $\Sigma_{Y|X_i}\succeq0$. Then the Blackwell redundancy $I_\cap$ in Eq.~\eqref{eq:red} may be computed by
restricting the optimization to Gaussian channels of output dimension $d_Q=\dim\bigcap_i\operatorname{range}(\eye-C_i)\le\min(d_Y,\min_i d_i)$:
\begin{equation}
I_{\cap}=\max_{\srcCh{Q}\in\CondGauss(d_Q)} I(Y;Q)
\quad\where\quad \srcCh{Q}\preceqY \srcCh{X_i}\;\;\forall i . #1
\end{equation}%
}

\newcommand{\thmUnionStatement}[1]{%
Suppose that each source--target pair $(X_i,Y)$ is jointly Gaussian, with
$\Sigma_{Y|X_i}\succ0$. Then the Blackwell union information $I_\cup$ in Eq.~\eqref{eq:red} may be
computed by restricting the optimization to Gaussian channels of output
dimension 
$d_Q=\dim\operatorname{range}\!\left(\sum_i(\eye-C_i)\right)\le\min(d_Y,\sum_i d_i)$:
\begin{equation}
I_{\cup}=\min_{\srcCh{Q}\in\CondGauss(d_Q)} I(Y;Q)
\quad\where\quad \srcCh{X_i}\preceqY \srcCh{Q}\;\;\forall i . #1
\end{equation}%
}

\newcommand{\thmTwoSourceStatement}[1]{%
For two sources with $C_1,C_2\succ0$, let $V$ and $\Lambda=\operatorname{diag}(\lambda_1,\dots,\lambda_{d_Y})$ solve
the generalized eigenvalue problem $C_2V=C_1V\Lambda$ with $V^\top C_1V=\eye$. 
Then
Eq.~\eqref{eq:union-convex} is solved by
$\Wcup=V^{-\top}\operatorname{diag}(\min(1,\lambda_j))V^{-1}$ and 
\begin{equation}
I_\cup=\frac12\bigl(I(Y;X_1)+I(Y;X_2)\bigr)+\frac14\sum_{j=1}^{d_Y}|\ln\lambda_j|\,. #1
\end{equation}%
}

\ifdoubleblind\else\iclrfinalcopy\fi
\title{Redundancy and synergy in multivariate\\Gaussians via the Blackwell order}
\ifdoubleblind\author{Anonymous authors}\else
\author{Artemy Kolchinsky\\
\emph{XOR Labs}\\
\emph{Barcelona Collaboratorium for Modelling and Predictive Biology, Barcelona, Spain}\\
\emph{Universal Biology Institute, University of Tokyo, Japan}
}
\fi
\begin{document}
\maketitle
\begin{abstract}
The goal of the partial information decomposition (PID) is to quantify the redundant and synergistic information that multiple sources provide about a target. PID has many applications in machine learning, neuroscience, and other fields, but defining and computing it for high-dimensional continuous systems remains challenging. Here, we define a PID for multivariate Gaussian systems based on the Blackwell order, which formalizes when one channel is more informative than another. We prove that Gaussian channels are optimal for extracting both redundant and union information, yielding an intuitive geometric interpretation and an efficient numerical algorithm for the PID.  Our union information and synergy coincide with the well-known BROJA measures, and we derive closed-form expressions for both in the case of two sources. We also argue that Blackwell redundancy (which differs from BROJA) is the only existing redundancy measure that satisfies a set of natural desiderata. 
We demonstrate the scalability of our method on systems with up to a thousand dimensions or sources. Our approach is illustrated on an optimal control problem, where it identifies redundant and synergistic interactions between sensor and memory.
\end{abstract}

\section{Introduction}

In many scientific domains, one seeks to understand how information about
a target variable $Y$ is distributed among a set of source
variables $X_1,\dots,X_n$. For example, in neuroscience, one may ask how
information about a stimulus is encoded across different brain regions~\citep{luppi2024information}; in
machine learning, one may ask how different features, layers, or
modalities combine to produce the model's output~\citep{liang2023quantifying}; and so on. To answer such questions, the \emph{partial
information decomposition} (PID)~\citep{williams2010nonnegative} aims to decompose  
the information that the sources have about the target into \emph{redundancy} (present in each individual
source), \emph{union information} (present in at least one
source), \emph{unique information} (non-redundant information in each source), and \emph{synergy} (present only in the joint outcome
of the sources). In general, the PID is inspired by an
analogy between information theory and set theory. For instance, if the information
in each source is imagined as a set, then redundancy is the size
of the intersection of these sets, while synergy is the difference between
the size of the joint set and that of the
union~\citep{williams_information_2011,kolchinskyNovelApproachPartial2022}. 

The PID requires an extension of standard information theory because existing quantities such as mutual information conflate synergy and redundancy~\citep{williams2010nonnegative}. 
However, development of this extension faces conceptual and practical challenges.

Conceptually, it has proven  
difficult to define the PID in a way that is
operationally meaningful and applicable to any number of
sources. Nonetheless, recent work has shown that these difficulties can be overcome
by making the information-set analogy 
formal; the only requirement is a mathematical condition that specifies when 
one channel is less informative about the target than another~\citep{kolchinskyNovelApproachPartial2022}. The \emph{Blackwell order} provides a principled answer: one channel is
less informative if it can be produced by stochastic
processing of the other or, equivalently, if it is no better on 
any decision problem~\citep{blackwell_equivalent_1953}. The resulting
``Blackwell PID'' leads to nonnegative measures of redundant, union, and synergistic information for any number of sources~\citep{kolchinskyNovelApproachPartial2022}. 
The practical challenge is that PID quantities, such as redundancy and synergy, may be difficult to compute, especially in 
continuous systems with many sources or high-dimensional variables.

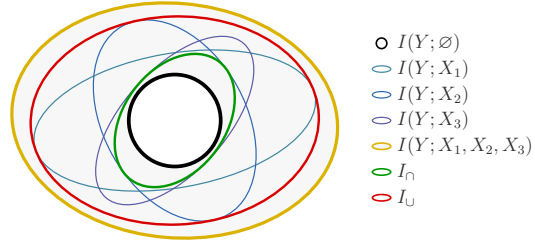
\begin{wrapfigure}{R}{0.5\textwidth}
\vspace{-12pt}
\centering
\resizebox{\linewidth}{!}{\begin{tikzpicture}[scale=1, every node/.style={font=\footnotesize}]
\definecolor{sourceone}{RGB}{62,134,160}
\definecolor{sourcetwo}{RGB}{46,104,178}
\definecolor{sourcethree}{RGB}{92,91,165}
\definecolor{totalinfo}{RGB}{220,178,0}
\tikzset{legendlabel/.style={anchor=west, inner sep=1pt, text=black!82,
        font=\large}}
\path[use as bounding box] (-5.70,-2.72) rectangle (5.70,2.72);
\begin{scope}[shift={(-1.55,0)}]
\begin{scope}[rotate=-5]
    \filldraw[fill=black!08, fill opacity=0.38, draw=totalinfo,
        line width=2pt] (0,0) ellipse (3.55 and 2.55);
\end{scope}
\begin{scope}[rotate=10]
    \draw[sourceone, line width=0.75pt] (0,0) ellipse (3.10 and 1.46);
\end{scope}
\begin{scope}[rotate=-61]
    \draw[sourcetwo, line width=0.75pt] (0,0) ellipse (2.35 and 1.55);
\end{scope}
\begin{scope}[rotate=48]
    \draw[sourcethree, line width=0.75pt] (0,0) ellipse (2.25 and 1.12);
\end{scope}
\begin{scope}[rotate=2.14]
    \draw[red!85!black, line width=1.5pt] (0,0) ellipse (3.136 and 2.264);
\end{scope}
\begin{scope}[rotate=53.6]
    \draw[green!60!black, line width=1.5pt] (0,0) ellipse (1.604 and 1.111);
\end{scope}
\filldraw[fill=white, draw=black, line width=2.3pt] (0,0) circle (1);
\end{scope}

\begin{scope}[shift={(2.95,1.68)}]
    \filldraw[fill=white, draw=black, line width=1pt] (0,0) circle (0.115);
    \node[legendlabel] at (0.30,0) {$I(Y;\varnothing)$};
    \draw[sourceone, line width=0.75pt] (0,-0.56) ellipse (0.20 and 0.07);
    \node[legendlabel] at (0.30,-0.56) {$I(Y;X_{1})$};
    \draw[sourcetwo, line width=0.75pt] (0,-1.12) ellipse (0.20 and 0.07);
    \node[legendlabel] at (0.30,-1.12) {$I(Y;X_{2})$};
    \draw[sourcethree, line width=0.75pt] (0,-1.68) ellipse (0.20 and 0.07);
    \node[legendlabel] at (0.30,-1.68) {$I(Y;X_{3})$};
    \filldraw[fill=black!10, fill opacity=0.38, draw=totalinfo,
        line width=1pt] (0,-2.24) ellipse (0.20 and 0.07);
    \node[legendlabel] at (0.30,-2.24) {$I(Y;X_1,X_2,X_3)$};
    \draw[green!60!black, line width=1pt] (0,-2.80) ellipse (0.20 and 0.07);
    \node[legendlabel] at (0.30,-2.80) {$I_{\cap}$};
    \draw[red!85!black, line width=1pt] (0,-3.36) ellipse (0.20 and 0.07);
    \node[legendlabel] at (0.30,-3.36) {$I_{\cup}$};
\end{scope}

\end{tikzpicture}}
\caption{\label{fig:1}PID geometry for multivariate Gaussians, shown here for three
sources $X_1,X_2,X_3$ and target $Y\in\mathbb{R}^2$. Each source $X_i$ is represented by a blue ellipse 
whose log-area is the mutual information $I(Y;X_i)$; the black circle
of unit area is the no-information channel. 
Redundancy $I_\cap$ corresponds to the largest (green) ellipse that is enclosed in every source and encloses the black circle; union information $I_\cup$ corresponds to the smallest (red) ellipse
enclosing all sources. Synergy is the difference of log-areas between the joint-source
ellipse (yellow) and the union ellipse.
\vspace{-20pt}
}
\end{wrapfigure}

As an important first step, we 
address these challenges in the setting of multivariate Gaussians, which 
play a fundamental role in time-series modeling, filtering, control theory, machine learning, and many other settings. Specifically, we demonstrate that the Blackwell order defines a principled PID for multivariate Gaussians. In doing so, we prove that Gaussian channels are optimal for extracting both redundant and union information. This leads to an intuitive geometric interpretation of PID redundancy, union, and synergy in terms of enclosing ellipsoids (Fig.~\ref{fig:1}). Moreover, we argue that the resulting measure of redundancy is the only existing one that satisfies a small set of natural desiderata (Sec.~\ref{sec:why-blackwell-redundancy}).  Extensions beyond Gaussians are also discussed.

Our formulation also has practical advantages. In particular, for any number of sources, we show that redundancy and union information can be efficiently computed by solving a convex optimization problem, which easily scales to a thousand dimensions or sources. In the special case of two sources, we uncover closed-form expressions for union information and synergy, allowing for very fast evaluation. 

We illustrate our approach on an optimal control
problem, where we quantify redundant and synergistic interactions between sensor and memory. Throughout, we measure information in units of nats. 
Proofs and additional material are provided in the Supplementary Material.


\section{Background: the Blackwell PID}
\label{sec:background}

We briefly review the Blackwell PID framework; see~\citet{kolchinskyNovelApproachPartial2022} for more 
details. 

Recall that, in set theory, the intersection of sets
$A_1,\dots,A_n$ is the largest set contained in each $A_i$; their union  
is the smallest set that contains each $A_i$. The sizes of the intersection
and the union can therefore be expressed variationally as
\begin{equation}
\Big|\bigcap_i A_i\Big|=\sup_{B:B\subseteq A_i\;\forall i}\,|B|,
\qquad\qquad
\Big|\bigcup_i A_i\Big|=\inf_{B:A_i\subseteq B\;\forall i}\,|B|.
\label{eq:set-analogy}
\end{equation}
These definitions require only a notion of set size $\vert \cdot \vert$ and the subset relation $\subseteq$. 

\citet{kolchinskyNovelApproachPartial2022} showed that
information-theoretic redundancy (i.e., intersection) and union information can be defined in a similar way. The analogue of a set is a \emph{channel} $\srcCh{A}$ (a conditional distribution of $A$ given the target), 
the analogue of set
size is the mutual information $I(Y;A)$, and the analogue of set inclusion
is an ordering relation $\srcCh{B}\preceqY\srcCh{C}$, which specifies when 
channel $B$ is ``less informative'' about $Y$ than channel $C$. Redundancy is the maximum of $I(Y;Q)$ over channels less informative than
every source, and union information is the minimum over channels
more informative than every source:
\begin{equation}
I_\cap:=\sup_{\srcCh{Q}:\srcCh{Q}\preceqY\srcCh{X_i}\;\forall i}I(Y;Q),
\qquad\qquad 
I_\cup:=\inf_{\srcCh{Q}:\srcCh{X_i}\preceqY\srcCh{Q}\;\forall i}I(Y;Q).
\label{eq:red}
\end{equation}
The
optimizing channels $\srcCh{Q^*}$ identify the content of redundant and union
information, not only their amount. 

Given the above definitions, the \emph{unique information} in source $i$ is $I(Y;X_i)-I_\cap\ge0$, the
\emph{excluded information} (union information absent from source $i$) is
$I_\cup-I(Y;X_i)\ge 0$, and \emph{synergy} is the information in the joint
sources beyond their union, $I_{\mathrm{syn}}:=I(Y;X_1,\dots,X_n)-I_\cup\ge 0$.

The definitions in Eq.~\eqref{eq:red} rely on a choice of ordering
relation. Perhaps the best-known and most principled choice is the
\emph{Blackwell order}~\citep{blackwell_equivalent_1953}. A channel
$\srcCh{B}$ is Blackwell less informative than $\srcCh{C}$, written
$\srcCh{B}\preceqY\srcCh{C}$, if there is a Markov kernel $\kappa$ from $C$
to $B$ such that $\srcCh{B}(E\vert y)=\int\kappa(E\vert c)\,\srcCh{C}(dc\vert y)$
for every measurable set $E$ and $P_Y$-almost every $y$. In other words,
$B$ can be simulated by \emph{garbling} the output of $C$, without further
access to $Y$. The Blackwell order has a fundamental operational characterization in decision theory~\citep{blackwell_equivalent_1953}:
$\srcCh{B}\preceqY\srcCh{C}$ iff an agent with $C$ can do at least as well as an agent
with $B$ on every decision problem
about $Y$.
Recently, the Blackwell order has attracted growing attention in machine learning applications~\citep{feng2026predictor,guo2026value,halder2025spuriousness,kaissis2024beyond}.

Combining the Blackwell order with Eq.~\eqref{eq:red}
gives the \emph{Blackwell PID}. Blackwell redundancy $I_\cap$ is the most
information one can have about $Y$ without being able to outperform any
single source on any decision problem; Blackwell union information $I_\cup$ is the
least information needed in order to perform at least as well as every
source. These measures are defined for any number of sources, uniquely
satisfy a natural set of axioms, and behave sensibly on various  
examples. Blackwell redundancy can also be formulated as an information-bottleneck-style objective, where one maximizes predictive information while hiding source identity~\citep{kolchinsky2024partial}.

Blackwell redundancy and union
information are formal duals, but their numerical values are not related by
inclusion--exclusion (M\"obius inversion). In fact, it has been argued that 
inclusion--exclusion should not apply to the PID, since it leads to inconsistencies~\citep{kolchinskyNovelApproachPartial2022}. 


\section{Other prior work}
\label{sec:existing-work}

An information-theoretic definition of union information and synergy, commonly referred to as
``BROJA'', was proposed by \citet{bertschinger2014quantifying}
and \citet{griffith2014quantifying}. The union information is defined as the smallest joint
information compatible with the source--target marginals,
\begin{equation}
I^{\mathrm{BR}}_\cup
:=\min_{\tilde P:\,\tilde P_{Y,X_i}=P_{Y,X_i}\ \forall i}
I_{\tilde P}(Y;X_1,\ldots,X_n).
\label{eq:broja-union}
\end{equation}
As usual, BROJA synergy is given by 
$I^{\mathrm{BR}}_{\mathrm{syn}}:=I(Y;X_1,\dots,X_n)-I^{\mathrm{BR}}_\cup$. 
Although defined in
different ways, Blackwell and BROJA union information/synergy were shown to be numerically equal:  $I_\cup=I^{\mathrm{BR}}_\cup$ and $I_{\mathrm{syn}}=I^{\mathrm{BR}}_{\mathrm{syn}}$~\citep{kolchinskyNovelApproachPartial2022}.  However, Blackwell redundancy is different from BROJA redundancy, which can be defined for two sources using inclusion--exclusion (see Sec.~\ref{app:two-source-iep}).


\citet{barrett_exploration_2015} first considered the PID in the Gaussian
setting. Given a scalar target and two sources, he showed that a broad class of
redundancy measures reduces to ``minimum mutual information'' (MMI) redundancy, defined as
$\min_i I(Y;X_i)$. 
\citet{goswami2024analytically} extend Barrett's MMI result to a broader class of scalar-target systems. Neither result generalizes to multivariate targets, because different sources can provide information about different directions of the target, which MMI cannot distinguish.

Recent work has considered BROJA PID for high-dimensional Gaussian and other continuous variables. \citet{venkatesh2023gaussian} develop a Gaussian-restricted BROJA optimizer, leaving the optimality of Gaussian solutions open. \citet{zhao2025partial} prove Gaussian optimality for two sources, introduce the faster Thin-PID solver, and use normalizing flows to extend the method to latent-Gaussian representations. \citet{pakman2021estimating} estimate BROJA unique information for low-dimensional continuous variables using copulas and variational optimization. \citet{liang2023quantifying} introduce estimators of BROJA PID quantities for high-dimensional multimodal data with non-Gaussian distributions. \citet{liardi2025scalable} propose a modified BROJA union information for multivariate targets, which preserves only the pairwise marginals between individual source and target components, and compute it for Gaussian systems.

Several alternative PIDs have also been proposed; see the following section and Sec.~\ref{app:why-blackwell}.

\section{Why Blackwell redundancy?}
\label{sec:why-blackwell-redundancy}

Before proceeding, we motivate Blackwell redundancy,
the most novel term in the Blackwell PID. 

Defining a measure of 
redundancy that works for more than two sources has proven remarkably difficult. In fact, many existing approaches 
are restricted to two
sources, including BROJA redundancy~\citep{bertschinger2014quantifying} and measures based on deficiency~\citep{venkatesh2022partial}, dependency constraints~\citep{kay2018exact}, information geometry~\citep{kay2024partial}, and conditional independence~\citep{lyu2026closed}. 
To our knowledge, besides Blackwell redundancy, there are five nontrivial 
redundancy measures that are defined for any number of sources and can be
evaluated for Gaussian variables: MMI~\citep{barrett_exploration_2015},
$I_\cap^{\min}$~\citep{williams2010nonnegative},
$I_\cap^{\mathrm{PM}}$~\citep{milzman2022signed},
$I_\cap^{\mathrm{ccs}}$~\citep{ince2017measuring}, and
$I_\cap^{\mathrm{sx}}$~\citep{ehrlich2024shared}.\footnote{The Gács--Körner
redundancy of \citet{griffith2014intersection}, based on the common information of \citet{gacs_common_1973}, is also defined for any number of
sources. We exclude it because it is trivial in generic Gaussian settings, vanishing whenever the sources have a nondegenerate joint Gaussian distribution (Sec.~\ref{app:why-blackwell}).}

Among these multi-source measures, Blackwell redundancy $I_\cap$ is the only one that 
satisfies two natural conditions. The first condition is a central axiom of the 
PID~\citep[Ch.~4]{williams_information_2011}; it states that redundancy should be monotonic, meaning $I_\cap$ does not increase
when additional sources are added (i.e., $A\cap B\cap C$ cannot be larger than $A\cap B$).
The second condition, usually termed the ``Independent Identity'' property in the literature~\citep{ince2017measuring}, states that redundancy should vanish when different sources depend on independent features
of the target. Specifically, $I_\cap=0$ if the target
$Y=(Y_1,Y_2)$ contains two independent features $Y_1\perp Y_2$, and each
source depends on only one of them, $P_{X_1|Y}=P_{X_1|Y_1}$ and
$P_{X_2|Y}=P_{X_2|Y_2}$. 

\begin{wraptable}{r}{0.35\textwidth}
\vspace{-1.5\baselineskip}
\caption{Redundancy in nats}
\label{tab:why-blackwell}
\centering
\setlength{\tabcolsep}{6pt}
\renewcommand{\arraystretch}{1.12}
\begin{tabular}{@{}lcc@{}}
\toprule
 & COPY & COPY+\\
\midrule
Expected                & $\match{0}$ & $\match{0}$\\
\midrule
$I_\cap$ (ours)         & \match{0} & \match{0}\\
MMI                     & $0.35$ & $0.35$\\
$I_\cap^{\mathrm{ccs}}$ & \match{0} & $0.01$\\
$I_\cap^{\min}$         & $0.19$ & $0.15$\\
$I_\cap^{\mathrm{PM}}$  & $0.35$ & $0.36$\\
$I_\cap^{\mathrm{sx}}$  & $0.35$ & $0.35$\\
\bottomrule
\end{tabular}
\end{wraptable}

We construct two Gaussian examples to test these properties (details in Sec.~\ref{app:why-blackwell}). Let
$Y_1,Y_2,\xi_1,\xi_2,\xi_3$ be independent standard Gaussians.
In the {COPY} example, the target has two independent
features, $Y=(Y_1,Y_2)$, and the two sources are $X_1=Y_1+\xi_1$ and
$X_2=Y_2+\xi_2$. Each source observes a different feature, so the second condition 
requires $I_\cap(Y;X_1,X_2)=0$.
The {COPY+} example adds a third source, $X_3=(Y_1+Y_2)/\sqrt2+\xi_3$.
By monotonicity, the three-source redundancy should be zero.
Table~\ref{tab:why-blackwell} reports redundancy. MMI, $I_\cap^{\min}$,
$I_\cap^{\mathrm{PM}}$, and $I_\cap^{\mathrm{sx}}$ already give positive
redundancy for COPY, violating the second condition. 
$I_\cap^{\mathrm{ccs}}$ gives zero for COPY but positive redundancy for
COPY+, violating monotonicity. Among the measures in the table, only Blackwell redundancy gives the
expected zero in both examples. 



\section{Blackwell PID for Gaussians}
\label{sec:gaussian-pid}
\label{sec:gaussian-channels}

We now present the Blackwell PID for multivariate Gaussians.
We suppose that the target
$Y\in\mathbb{R}^{d_Y}$ is multivariate Gaussian with mean $\mu_Y$ and
positive-definite covariance $\Sigma_Y\succ0$. We also suppose  that each source
$X_i\in\mathbb{R}^{d_i}$ is jointly Gaussian with the target, so that
$\srcCh{X_i}(\cdot|y)=\mathcal{N}\big(\mu_{X_i}+\Sigma_{X_iY}\Sigma_Y^{-1}(y-\mu_Y),\Sigma_{X_i\vert Y}\big)$. 
Since redundancy and union information depend only on the source--target
marginals, it is enough that each pair $(X_i,Y)$ is jointly Gaussian;
synergy additionally involves the joint distribution of all sources and the target.

In our first result, we show that the Blackwell order can be written as a matrix inequality
between conditional covariance matrices. Here $\leL$ is the Loewner order, i.e., 
$V\leL W$ means that $W-V$ is positive semidefinite.
\begin{thm}
\label{thm:posterior-gaussian-blackwell}
\thmBlackwellStatement{\label{eq:gaussian-blackwell}}
\end{thm}
In this result, we do not assume joint Gaussianity,
only that the target may be written as a (possibly nonlinear) function of
each source plus independent Gaussian noise. In Sec.~\ref{sec:beyond-gaussians}, we 
use this to discuss generalizations beyond Gaussian assumptions. 
For jointly Gaussian channels, Theorem~\ref{thm:posterior-gaussian-blackwell}
reduces to a classical result by~\citet{hansen1974comparison}; see also
\citet[Theorem~2]{venkatesh2022partial}. 

Theorem~\ref{thm:posterior-gaussian-blackwell} has a simple geometric
interpretation. Associate with each channel $\srcCh{A}$ the ellipsoid
$\mathcal{E}_A:=\{z\in\mathbb{R}^{d_Y}:z^\top\Sigma_{Y|A}z\le1\}$. Then the
Blackwell order is ellipsoid inclusion: $\srcCh{A}\preceqY\srcCh{B}$ if and
only if $\mathcal{E}_A\subseteq\mathcal{E}_B$. The no-information ellipsoid 
$\mathcal{E}_\varnothing:=\{z:z^\top\Sigma_Yz\le1\}$ represents a channel that has no information about the target, and it is
contained in every $\mathcal{E}_A$. Mutual information across the channel is a log-volume
ratio, $I(Y;A)=\frac12\ln(|\Sigma_Y|/|\Sigma_{Y|A}|)
=\ln(\mathrm{Vol}\,\mathcal{E}_A/\mathrm{Vol}\,\mathcal{E}_\varnothing)$.
Fig.~\ref{fig:1} illustrates this
geometry for a two-dimensional target in whitened coordinates, with areas
normalized so that $\mathcal{E}_\varnothing$ has unit area.

In the following, it will be 
convenient to work in whitened target coordinates, using the
\emph{normalized posterior covariance},
\begin{align}\Cmat{A}:=\Sigma_Y^{-1/2}\Sigma_{Y\vert A}\Sigma_Y^{-1/2}\,.\end{align}
This matrix obeys 
$I(Y;A)=-\tfrac12\ln|\Cmat{A}|$. For the source channels $\srcCh{X_i}$, we abbreviate $C_i:=\Cmat{X_i}$. 
By
Theorem~\ref{thm:posterior-gaussian-blackwell},
$\srcCh{A}\preceqY\srcCh{B}$ iff $\Cmat{A}\geL\Cmat{B}$. 

Observe that we have $0\preceq\Cmat{A}\preceq\eye$ always. 
Conversely, every matrix $0\preceq R\preceq\eye$ is the normalized
posterior covariance of some Gaussian channel (see
Sec.~\ref{app:gaussian-comparison}),
\begin{equation}
\srcCh{Q_R}(\cdot\vert y)
=\mathcal N\!\left((\eye-R)^{1/2}\Sigma_Y^{-1/2}(y-\mu_Y),\,R\right).
\label{eq:canonical-posterior-channel}
\end{equation}

\subsection{Redundancy}
\label{sec:redundancy}

We now consider Blackwell redundancy for multivariate Gaussians.  In the following, we write $\CondGauss(d)$
for the set of Gaussian channels $\srcCh{Q}$ with $Q\in\mathbb{R}^{d}$,
allowing singular conditional covariances, and $\psd{d}$ for the set of
positive semidefinite $d\times d$ matrices.

The definition of Blackwell redundancy~\eqref{eq:red}
optimizes over all channels, including non-Gaussian and
infinite-dimensional ones. Our next result shows that the optimizer may be taken as Gaussian.
\begin{thm}
\label{thm:gaussian-optimality-redundancy}
\thmRedundancyStatement{\label{eq:redThm}}
\end{thm}
Combining Theorems~\ref{thm:posterior-gaussian-blackwell} and \ref{thm:gaussian-optimality-redundancy} allows us to express redundancy as an optimization over covariance matrices. It is convenient to work with the inverse matrix
$M:=C_Q^{-1}$, which satisfies
$I(Y;Q)=\frac12\ln|M|$ and $M\succeq\eye$. The Blackwell constraint
$\Cmat{Q}\succeq C_i$ then reads $C_i^{1/2}MC_i^{1/2}\preceq\eye$, valid even when $C_i$ is singular. $I_\cap$ can then be computed by solving the convex matrix optimization problem:
\begin{equation}
I_{\cap}=\sup_{M\in\psd{d_Y}}\frac{1}{2}\ln\left|M\right|
\quad\where\quad \eye\preceq M,\qquad C_i^{1/2}MC_i^{1/2}\preceq\eye\quad\forall i\,.
\label{eq:bl2}
\end{equation}

When $\sum_i C_i\succ0$, Eq.~\eqref{eq:bl2} has a unique optimizer $\Mcap$ and the redundancy is finite.
The optimal redundant channel is given by
Eq.~\eqref{eq:canonical-posterior-channel}: 
\begin{align}\srcCh{Q^*_\cap}(\cdot|y)=\mathcal{N}\big((\eye-\Mcap^{-1})^{1/2}\Sigma^{-1/2}_{Y}(y-\mu_{Y}),\Mcap^{-1}\big).
\end{align}
This channel specifies the {content} of the redundant information.
If $\sum_i C_i$ is singular, then every source observes a
common target direction with infinite precision, so $I_\cap=+\infty$.

Recall the geometric interpretation of the Blackwell order as  ellipsoid
inclusion. From this interpretation, $I_\cap$ corresponds to the largest ellipsoid that encloses the
no-information ellipsoid $\mathcal{E}_\varnothing$ and is enclosed in every source ellipsoid 
(Fig.~\ref{fig:1}). This is the Gaussian analogue of set intersection.

When $Y$ is scalar, the $C_i$ are scalars and Eq.~\eqref{eq:bl2} gives the well-known ``minimum mutual
information'' (MMI) redundancy $I_\cap=\min_i
I(Y;X_i)$~\citep{barrett_exploration_2015}. This equivalence does not hold for the more interesting case of  multivariate targets, where different sources may provide information about different target directions.

\textbf{The no-information constraint.}\label{sec:not-a-lattice} The constraint $\eye\preceq M$ in Eq.~\eqref{eq:bl2} reflects the fact that no channel can produce more posterior uncertainty about the target than observing nothing. 
Somewhat surprisingly, this constraint is not implied by the source constraints. In terms of the geometric picture in Fig.~\ref{fig:1}, the largest ellipsoid contained in all source ellipsoids need not contain the no-information ellipsoid, even though each source ellipsoid does. This is related to the fact 
that neither the Blackwell order nor
the Loewner order is a lattice~\citep{bertschinger2014blackwell,kadison1951order}, so common lower bounds need not have a
greatest element.

Also, for two sources, Blackwell redundancy and union information need not satisfy the inclusion--exclusion formula $I_\cap=I(Y;X_1)+I(Y;X_2)-I_\cup$~\citep{kolchinskyNovelApproachPartial2022}. However, in Gaussian systems with $C_1,C_2\succ0$, this formula holds precisely when the no-information constraint $\eye\preceq M$ is non-binding, i.e., when dropping it does not change the optimum of Eq.~\eqref{eq:bl2}. In that case, the optimal matrices of the redundancy and union problems obey a matrix form of inclusion--exclusion,
\begin{equation}
\Mcap^{-1}+\Wcup=C_1+C_2 \,,
\label{eq:matrix-iep}
\end{equation}
where $\Wcup$ is the union optimizer of Sec.~\ref{sec:union}. In this case, the redundant channel can be read off from the union channel; see Sec.~\ref{app:two-source-iep} for details. 

\subsection{Union information and synergy}
\label{sec:union}

Next, we turn to Blackwell union information for multivariate Gaussians. We first show that, as for redundancy, it suffices to
optimize over Gaussian channels. Without loss of generality, we assume $\Sigma_{Y|X_i}\succ0$ for all $i$; if some $\Sigma_{Y|X_i}$ is singular, then $I_\cup\ge I(Y;X_i)=+\infty$. 
\begin{thm}
\label{thm:gaussian-optimality-union}
\thmUnionStatement{\label{eq:unqThm}}
\end{thm}
Writing $W:=\Cmat{Q}\equiv \Sigma_Y^{-1/2}\Sigma_{Y\vert Q}\Sigma_Y^{-1/2}$, the Blackwell constraint
$\srcCh{X_i}\preceqY\srcCh{Q}$ becomes $W\preceq C_i$, so union
information is given by the convex matrix optimization problem:
\begin{equation}
I_{\cup}=\min_{W\in\psd{d_Y}}-\frac{1}{2}\ln\left|W\right|
\quad\where\quad W\leL C_i\quad\forall i.
\label{eq:union-convex}
\end{equation}
The optimizer $\Wcup$ is unique. An
optimal union channel is obtained from
Eq.~\eqref{eq:canonical-posterior-channel}, now with $R=\Wcup$:
\begin{equation}
\srcCh{Q^*_\cup}(\cdot|y)=\mathcal{N}\big((\eye-\Wcup)^{1/2}
\Sigma^{-1/2}_{Y}(y-\mu_{Y}),\,\Wcup\big).
\label{eq:union-optimal-channel}
\end{equation}

{Eq.~\eqref{eq:union-convex} is the geometric dual of the redundancy problem. In terms of the geometric interpretation visualized in Fig.~\ref{fig:1}, union information corresponds to the smallest
ellipsoid that encloses all source ellipsoids}.  This may be interpreted as  the
Gaussian analogue of set union. 

Synergy is the information available from the joint
sources beyond the union channel,
$I_{\mathrm{syn}}=I(Y;X_1,\ldots,X_n)-I_\cup\ge0$. If the full joint distribution of $(Y,X_1,\dots,X_n)$ is
Gaussian, then $I_{\mathrm{syn}}=-\frac12\ln|\Cmat{X_1,\ldots,X_n}|-I_\cup$,
the difference between the log-volumes of the joint-source ellipsoid and
the union ellipsoid (Fig.~\ref{fig:1}).

Recall that Blackwell and BROJA union information/synergy are numerically equal. Using the coupling construction of
\citet[Appendix~G]{kolchinskyNovelApproachPartial2022} with Gaussian garbling kernels,
Theorem~\ref{thm:gaussian-optimality-union} also yields a jointly Gaussian optimizer for BROJA union information
with any number of sources. 
\citet{zhao2025partial} derived a similar
result in the case of two sources.  In addition, our results imply that the standard BROJA
optimization over high-dimensional joint distributions~\citep{venkatesh2023gaussian,zhao2025partial} can be replaced by an optimization over a
single $d_Y\times d_Y$ matrix.

Remarkably, for two sources, the Blackwell union problem can be solved in closed
form, removing the need for any optimization (see Sec.~\ref{app:two-source}).
\begin{thm}
\label{thm:two-source-union-closed-form}
\thmTwoSourceStatement{\label{eq:two-source-union}}
\end{thm}
This also implies a closed-form solution for BROJA union/synergy in the case of two sources.

For a scalar target, Eq.~\eqref{eq:union-convex} gives
$I_\cup=\max_i I(Y;X_i)$, so union information is the information in the
strongest source $X_{i^*}$, and synergy is the corresponding conditional mutual information
$I(Y;X_1,\dots,X_n\vert X_{i^*})$~\citep{barrett_exploration_2015}. In the case of a multidimensional target, however, union information and synergy need not reduce to these simple mutual information expressions.




\section{Numerical optimization}
\label{sec:numerical}

\begin{figure}[b!]
\centering
\includegraphics[width=\textwidth]{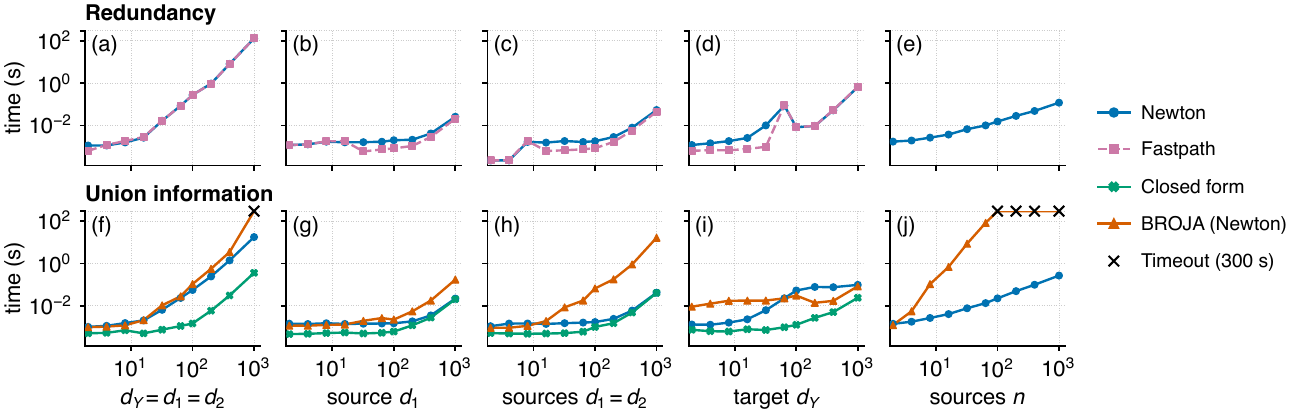}
\caption{\label{fig:pid-value-timing}\label{fig:redundancy-total-timing}\label{fig:union-total-timing}
Runtimes across five matched trials for redundancy (a)--(e) and union information (f)--(j).
Crosses mark a $300$-s timeout in at least one BROJA trial.
Methods: \textit{Newton}: optimization of Eqs.~\eqref{eq:bl2} and~\eqref{eq:union-convex} using Newton--PCG;
\textit{Fastpath}: two-source redundancy shortcut with Newton fallback;
\textit{Closed form}: two-source union formula from Theorem~\ref{thm:two-source-union-closed-form}.
\textit{BROJA (Newton)}: Newton--PCG implementation of the BROJA source-covariance
optimization~\citep{venkatesh2023gaussian,zhao2025partial}.
Panels (a,f) vary source and target dimensions together;
(b,g) vary the first source dimension $d_1$ ($d_Y=d_2=8$);
(c,h) vary both source dimensions $d_1=d_2$ ($d_Y=8$);
(d,i) vary the target dimension $d_Y$ ($d_1=d_2=64$);
(e,j) vary the source count $n$ ($d_Y=d_i=8$).
Details in Sec.~\ref{sm:sec:experiments}.
}
\end{figure}

\label{sec:numerical-algorithms}

We now describe efficient algorithms for computing the Gaussian Blackwell PID and evaluate their performance. See Sec.~\ref{sm:sec:optimization} for details. 

Gaussian redundancy and union, Eqs.~\eqref{eq:bl2}
and~\eqref{eq:union-convex}, are defined in terms of convex log-determinant problems in a
single $d_Y\times d_Y$ matrix, so they can be solved with off-the-shelf optimization packages such as CVXPY~\citep{diamond2016cvxpy}. However, we obtained better
performance by combining a Newton--PCG optimizer and log-barrier methods.
Our implementation uses JAX to just-in-time compile the full optimization loop.
It returns certified lower and upper bounds on the globally optimal information
value and stops when their gap is at most $\epsilon_{\rm abs}$.


For the algorithms below, assume $C_i\succ0$ after the preprocessing
described in Sec.~\ref{sm:sec:optimization}, which handles singular cases.
The size of our optimization problem is set by the
target dimension $d_Y$.  When the sources are
low-dimensional relative to the target, each source
observes only a low-dimensional subspace of the target and our optimization problems can be
reparameterized in terms of smaller matrices. Specifically, for
{redundancy}, the constraints $\eye\preceq M\preceq C_i^{-1}$ imply that 
$M-\eye\succeq0$ and that its range is contained in the subspace
$\bigcap_i\operatorname{range}(C_i^{-1}-\eye)$, whose dimension is the
$d_Q\le\min(d_Y,\min_i d_i)$ of Theorem~\ref{thm:gaussian-optimality-redundancy}.
Writing $P\in\mathbb{R}^{d_Y\times d_Q}$ for an orthonormal basis of this
subspace, every feasible matrix is $M=\eye+PKP^\top$ with $K\in\psd{d_Q}$,
and Eq.~\eqref{eq:bl2} becomes an optimization over smaller matrices $K$. 
For
{union}, write the whitened channels as $X_i=H_iY+N_i$, 
so that
source $i$ depends on the target only through the projection of
$Y$ onto the row space of $H_i$. Each source channel therefore depends on the target
only through the subspace $\operatorname{range}\big(\sum_iH_i^\top H_i\big)$, whose dimension is the $d_Q\le\min(d_Y,\sum_i d_i)$ of Theorem~\ref{thm:gaussian-optimality-union}. 
Directions orthogonal to this subspace can be
dropped: if $P_\cup\in\mathbb{R}^{d_Y\times d_Q}$ is an orthonormal basis of
the observed subspace, then $P_\cup^\top Y$ is a $d_Q$-dimensional 
target and $H_iY=H_iP_\cup\,(P_\cup^\top Y)$, so replacing each $H_i$ by
$H_iP_\cup$ gives a
$d_Q$-dimensional problem with the same union information. 

{In the case of two sources}, union and redundancy can often be computed faster by using  the generalized eigenvalues $\lambda_j$ of the pair
$(C_2,C_1)$. 
{To compute {union} information for two sources}, we evaluate the
closed form of Theorem~\ref{thm:two-source-union-closed-form},  $ I_\cup=-\frac12(\ln|C_1|+\sum\nolimits_j\ln\min(1,\lambda_j))$, 
which requires only the eigenvalues, with no eigenvectors and no
reconstruction of the optimal matrix $\Wcup$. 
 {To compute {redundancy} for two sources}, we use the logic of
Eq.~\eqref{eq:matrix-iep}: from the closed-form union optimizer $\Wcup$ of
Theorem~\ref{thm:two-source-union-closed-form}, we form the candidate
$R=C_1+C_2-\Wcup$. If $R\preceq\eye$, the no-information constraint is not
binding, so $\Mcap=R^{-1}$ (Sec.~\ref{app:two-source-iep}) and $I_\cap=I(Y;X_1)+I(Y;X_2)-I_\cup$. Otherwise,
we solve Eq.~\eqref{eq:bl2} with the Newton--PCG optimizer. 

\label{sec:numerical-experiments}

Figure~\ref{fig:pid-value-timing} shows runtimes on random Gaussian
channels $X_i=H_iY+N_i$, with dimensions and source counts up to $10^3$
and matched instances across methods.
For redundancy, we compare Newton and Fastpath. For union, existing
Gaussian BROJA solvers optimize over source-noise covariances, such as a
$d_1\times d_2$ cross-noise covariance for two
sources~\citep{venkatesh2023gaussian,zhao2025partial}, whereas we
optimize a single $d_Y\times d_Y$ matrix. We compare against BROJA (Newton),
which solves the source-space problem with the same Newton--PCG driver;
the Thin-PID optimizer of \citet{zhao2025partial} is included in
Fig.~\ref{sm:fig:pid-value-timing}.

Our optimizer provides a certified solution that is within 
$\epsilon_{\rm abs}=0.01$ nats of global optimality. In contrast, the competing methods do not
compute an optimality certificate. We run them until their union
objective reaches our certified upper bound on the same
instance. 
(This gives them an advantage, since they are not required to certify their results.)
\section{Example: optimal control}
\label{sec:optimal-control}

We now illustrate our approach using a simple optimal control model. Details are found in Sec.~\ref{app:optimal-control}.

We consider a 
linear--quadratic--Gaussian (LQG) controller 
 that confines a two-dimensional particle $\mathbf r \in\mathbb R^2$. The particle diffuses as  $ d\mathbf r=(J\mathbf r+\mathbf u)\,dt+d\mathbf W$, where $\mathbf u$ is the applied control and $J$ represents an elliptical trap with rotational drift. At each sampling
interval $\Delta t$, the controller makes a noisy position measurement $\mathbf s_t=\mathbf r_t+\mathbf n_t$, where $\mathbf n_t \sim \mathcal{N}(0,\nu S)$ for  an anisotropic covariance matrix $S$ and noise scale $\nu$. This measurement is combined with the memory $\mathbf m_t$, the
Kalman prediction of the current position from earlier measurements, to compute the applied control $\mathbf u_t=-F\big[L\mathbf s_t+(I-L)\mathbf m_t\big]+\boldsymbol\eta_t$, 
where $L$ is the Kalman gain, $F$ is the feedback gain, and
$\boldsymbol\eta_t$ is actuator noise. The memory is then updated.

\begin{wrapfigure}{R}{0.48\textwidth}
\vspace{-10pt}
 \centering
 \includegraphics[width=\linewidth]{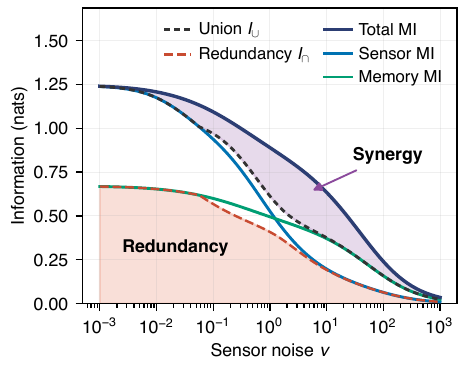}
 \caption{PID analysis of how the applied control (target) integrates two sources, sensor measurements and memory, as a function of sensor noise
 $\nu$. Shading shows
 redundancy $I_\cap$ and synergy $I(Y;X_1,X_2)-I_\cup$. For details, see Sec.~\ref{app:optimal-control}.}
 \label{ctrl:figure}
 \vspace{-10pt}
\end{wrapfigure}

To study information processing in the controller, we consider the current sensor reading $X_1=\mathbf s_t$ and the memory
$X_2=\mathbf m_t$ as two sources; we consider the control force $Y=\mathbf u_t$ as
the target. Mutual information under the steady-state Gaussian distribution quantifies how informative each source (or both
together) is about the control, but not if and how the controller combines them in computing $\mathbf u$.

To make progress, we apply the Blackwell PID framework. In Figure~\ref{ctrl:figure}, we sweep the sensor noise strength $\nu$ at a fixed sampling interval,
recomputing the optimal controller and the Blackwell PID quantities. 
At low sensor noise, the controller mainly uses the current sensor measurement. Memory, which propagates earlier
readings through the known dynamics, predicts the current sensor reading well but adds little further information. Memory is then Blackwell less informative than
the sensor, its information is entirely redundant, and synergy is negligible. As $\nu$ grows, the sensor becomes noisy and the controller shifts weight toward memory. The two sources now carry different information:
the sensor captures fluctuations that memory cannot anticipate, while
memory contains evidence accumulated over earlier measurements. Because the control depends on the
combination of the two, part of the information about it is available
only jointly. Redundancy falls and synergy rises, peaking at $0.31$ nats
($39\%$ of the total) near $\nu\simeq2.6$. 
At high sensor noise, the sensor is nearly ignored and
the control follows memory. The sensor is now Blackwell less informative
than memory, so its information is entirely redundant and the remainder
of the union information is unique to memory.

\section{Beyond Gaussians}
\label{sec:beyond-gaussians}

In this paper, we focus on multivariate Gaussians, the fundamental starting point for continuous systems.
Extensions to non-Gaussian systems are a natural next direction. As we note here, some
of our results already apply in this setting. (See Secs.~\ref{app:redundancy-optimality} and~\ref{app:union-optimality} for more details.)

Theorem~\ref{thm:posterior-gaussian-blackwell}, which relates Blackwell and matrix orders, does not require Gaussian
sources or targets, only that the target can be predicted from each source
with additive homoscedastic Gaussian noise:
\begin{equation}
Y=m_i(X_i)+Z_i,\qquad Z_i\sim\mathcal N(0,\Sigma_{Y|X_i}),\qquad Z_i\perp X_i
\quad\forall i,
\label{eq:posterior-gaussian-assumption}
\end{equation}
where $m_i$ may be nonlinear. For the formulas below, assume finite target
covariance $\Sigma_Y\succ0$ and $\Sigma_{Y|X_i}\succ0$ for every $i$.
Let $I_\cap^{\mathrm G}$ and $I_\cup^{\mathrm G}$ denote the values of the
Gaussian programs~\eqref{eq:bl2} and~\eqref{eq:union-convex}, evaluated using
the normalized conditional covariances $C_i=\Sigma_Y^{-1/2}\Sigma_{Y|X_i}\Sigma_Y^{-1/2}$.


Let $\tilde Y$ be Gaussian with the same mean and covariance as $Y$, and let the KL divergence 
$D(P_Y\Vert P_{\tilde Y})\ge0$ quantify the target's
deviation from Gaussianity. Blackwell redundancy obeys the upper bound
$I_\cap\le I_\cap^{\mathrm G}-D(P_Y\Vert P_{\tilde Y})\le I_\cap^{\mathrm G}$.
Blackwell union information extends directly:
$I_\cup=I_\cup^{\mathrm G}-D(P_Y\Vert P_{\tilde Y})\le I_\cup^{\mathrm G}$. As usual, synergy can then be computed as
$I_{\mathrm{syn}}=I(Y;X_1,\dots,X_n)-I_\cup$.
Finally, if $Y$ is Gaussian, then $D(P_Y\Vert P_{\tilde Y})=0$, and redundancy and union
information are given exactly by $I_\cap^{\mathrm G}$ and $I_\cup^{\mathrm G}$,
even if the sources do not have Gaussian marginals.


A complementary and more general approach, which we do not pursue here, is to learn invertible transformations into a Gaussian latent space, as in Flow-PID~\citep{zhao2025partial}.
Separate invertible transformations of the sources and the target preserve
mutual information and Blackwell comparisons, so our redundancy and union
programs apply exactly if the transformed source--target pairs are Gaussian in the latent space; otherwise, they provide approximations. PIDReg~\citep{ma2025explainable}
also applies Gaussian PID to learned representations, without requiring
invertible encoders.



\section{Discussion}
\label{sec:discussion}

We have shown that the Blackwell PID admits an exact and tractable
characterization for arbitrarily many Gaussian sources and target
dimensions, leading to a simple geometric interpretation and efficient
numerical methods. This provides a
principled PID framework for characterizing how information is encoded across
distributed representations. 

We are also able to identify optimal channels that 
extract the content of redundant and union information. 
Although we have not explored this here, such analyses could be useful in various data-driven applications. For example, in high-dimensional neuroimaging
data~\citep{luppi2024information}, one may consider redundant modes shared
across brain regions or timepoints. 


Another direction is to study how
the organization of information changes during learning. Deep linear networks~\citep{saxe2013exact} provide a controlled setting in
which to follow redundancy and synergy across training
trajectories, and the optimal channels could reveal whether learning
causes different units to represent the same target features or to
specialize in different ones. 

Finally, it is interesting to explore the use of our 
PID quantities as components of learning
objectives. Penalizing redundancy
between groups of units or modalities would push them to encode different
target features, favoring disentangled or complementary representations,
whereas rewarding redundancy would favor robustness to the loss of
individual units. Similarly, penalizing synergy would favor representations whose information
can be read out from individual components.  We also note that PID-based objectives have been proposed as local goal functions
for individual neurons~\citep{wibral2017partial,makkeh2025general,schneider2025what};
our approach could be used to extend such objectives to continuous, multivariate
representations.


\section*{Reproducibility statement}
All proofs are given in the Supplementary Material, Sec.~\ref{sec:sm}. 
The optimal-control model is described in Sec.~\ref{app:optimal-control}. 
Details of the optimization algorithm, including the Newton--PCG solver, are provided in
Sec.~\ref{sm:sec:optimization}. 
\ifdoubleblind
\else
The solver and
current timing benchmark code are available in the
\href{https://github.com/artemyk/gaussian-pid}{Gaussian-PID repository}.
Code reproducing the remaining figures and tables will be released with
the paper.
\fi

\section*{AI use statement}
AI assistance, including OpenAI Codex and Anthropic Claude, was used to assist with mathematical derivations, implementation of numerical solvers, generation of figures, and manuscript editing. 
The final version of the manuscript, including all mathematical proofs, was verified by the \ifdoubleblind\else author\fi. 
Responsibility for all scientific claims, results, and text rests with the \ifdoubleblind\else author\fi.

\bibliographystyle{iclr2027_conference}
\bibliography{writeup,setup-references}
\clearpage
\appendix
\begin{center}
{\Large\bf Supplementary Material}
\end{center}
\vspace{0.5\baselineskip}
\setcounter{section}{0}
\renewcommand{\thesection}{S\arabic{section}}
\renewcommand{\theHsection}{SM.\arabic{section}}
\setcounter{equation}{0}
\setcounter{figure}{0}
\setcounter{table}{0}
\renewcommand{\theequation}{S\arabic{equation}}
\renewcommand{\thefigure}{S\arabic{figure}}
\renewcommand{\thetable}{S\arabic{table}}
\renewcommand{\theHequation}{SM.\arabic{equation}}
\renewcommand{\theHfigure}{SM.\arabic{figure}}
\renewcommand{\theHtable}{SM.\arabic{table}}
\section{Theoretical results}
\label{sec:sm}
\subsection{Gaussian posteriors and Blackwell order}
\label{app:gaussian-comparison}

Before proceeding to our proof of Theorem~\ref{thm:posterior-gaussian-blackwell}, we introduce a helpful lemma.
Here and below, $\Sigma_{Y|B=b}$ denotes the covariance of $P_{Y|B=b}$, and
$\Sigma_{Y|B}:=\int\Sigma_{Y|B=b}\,P_B(db)$ its average.

\begin{lem}
\label{lem:garbling-covariance}
Suppose $\srcCh{A}\preceqY\srcCh{B}$, the posteriors of $Y$ given
$A$ and $B$ have finite second moments almost surely, and
$\Sigma_{Y|A=a}=\Sigma_{Y|A}$ is constant. Then
$\Sigma_{Y|A}\succeq\Sigma_{Y|B}$.
\end{lem}
\begin{proof}
Couple $(Y,B,A)$ through the garbling kernel so that $Y-B-A$ is a Markov
chain, and write $m_B(b)=\int y\,P_{Y|B=b}(dy)$. For $P_A$-almost every
$a$, the conditional law of total covariance gives
\begin{equation}
\Sigma_{Y|A}
=\int\Sigma_{Y|B=b}\,P_{B|A=a}(db)
+\operatorname{Cov}(m_B(B)|A=a)
\succeq\int\Sigma_{Y|B=b}\,P_{B|A=a}(db).
\end{equation}
Averaging over $A$ proves the claim.
\end{proof}

\noindent\textbf{Theorem~\ref{thm:posterior-gaussian-blackwell} (restated).}
\begin{itshape}
\thmBlackwellStatement{}
\end{itshape}

\begin{proof}
Suppose first that $\srcCh{A}\preceqY\srcCh{B}$. The Gaussian posteriors
have finite second moments and constant covariances, so
Lemma~\ref{lem:garbling-covariance} gives $\Sigma_{Y|A}\succeq \Sigma_{Y|B}$.

Conversely, suppose $D:=\Sigma_{Y|A}-\Sigma_{Y|B}\succeq0$, and let
$V_A=m_A(A)$ and $V_B=m_B(B)$. 
Because
$P_{Y|A=a}$ depends on $a$ only through $m_A(a)$, we have the Markov chain $Y - V_A - A$ and so $\srcCh{A}\preceqY \srcCh{V_A}$.
Since $V_B$ is a function of $B$,
we also have $P_{V_B|Y}\preceqY\srcCh{B}$. 
It therefore suffices to show
$P_{V_A|Y}\preceqY P_{V_B|Y}$. The two posterior models give the same
target distribution,
\begin{equation}
P_{V_A}*\mathcal N(0,\Sigma_{Y|A})=P_Y=P_{V_B}*\mathcal N(0,\Sigma_{Y|B}).
\end{equation}
For a distribution $P$ on $\mathbb R^{d_Y}$, let
$\widehat P(t):=\int e^{\mathrm{i}t^\top v}\,P(dv)$ denote its
characteristic function. The
characteristic function of a convolution is the product of the
characteristic functions, and that of $\mathcal N(0,\Sigma)$ is
$e^{-\frac12t^\top\Sigma t}$. Hence, for every $t\in\mathbb R^{d_Y}$,
\begin{equation}
\widehat P_{V_A}(t)\,e^{-\frac12t^\top\Sigma_{Y|A}t}
=\widehat P_{V_B}(t)\,e^{-\frac12t^\top\Sigma_{Y|B}t}.
\end{equation}
Dividing by the strictly positive factor $e^{-\frac12t^\top\Sigma_{Y|B}t}$ gives $\widehat P_{V_B}(t)=\widehat P_{V_A}(t)\,e^{-\frac12t^\top D t}$. 
The right-hand side is the characteristic function of
the distribution $P_{V_A}*\mathcal N(0,D)$. 
A distribution is uniquely determined by its characteristic
function, so $P_{V_B}=P_{V_A}*\mathcal N(0,D)$.
Now draw independent $V_A\sim P_{V_A}$, $U\sim\mathcal N(0,D)$, and
$Z\sim\mathcal N(0,\Sigma_{Y|B})$, and set $V_B=V_A+U$ and $Y=V_B+Z$. This
reproduces the joint distributions of $(Y,V_A)$ and $(Y,V_B)$, since
$U+Z\sim\mathcal N(0,\Sigma_{Y|A})$, and forms the Markov chain
$Y- V_B- V_A$, so $P_{V_A|Y}\preceqY P_{V_B|Y}$. Hence
$\srcCh{A}\preceqY P_{V_A|Y}\preceqY P_{V_B|Y}\preceqY\srcCh{B}$.
\end{proof}

\paragraph{Gaussian targets and canonical channels.}
In the proofs of Theorems~\ref{thm:gaussian-optimality-redundancy}--\ref{thm:two-source-union-closed-form}, assume that
\(Y\in\mathbb{R}^{d_Y}\) is Gaussian with mean \(\mu_Y\) and
positive-definite covariance \(\Sigma_Y\succ0\). We use the normalized posterior covariances
$\Cmat{A}$ and $C_i:=\Cmat{X_i}$ of Sec.~\ref{sec:gaussian-pid}, allowing
degenerate channels, with the convention $-\ln|C|=+\infty$ for singular
positive-semidefinite $C$.

For the canonical channel~\eqref{eq:canonical-posterior-channel}, the
whitened target $\Sigma_Y^{-1/2}(Y-\mu_Y)$ and $Q_R$ are jointly Gaussian with
$\operatorname{Cov}(Q_R)=\eye$ and cross-covariance $(\eye-R)^{1/2}$, so
$\Cmat{Q_R}=\eye-(\eye-R)=R$ and $I(Y;Q_R)=-\frac12\ln|R|$, which is
$+\infty$ when $R$ is singular (directions in $\ker R$ are observed without
noise). Output coordinates in $\ker(\eye-R)$ are pure noise. Discarding
them, or padding with independent noise, gives an equivalent channel of
any output dimension $d\ge\operatorname{rank}(\eye-R)$.

\subsection{Gaussian optimality of redundancy}
\label{app:redundancy-optimality}

We first prove an upper bound under the more general posterior assumption
of Sec.~\ref{sec:beyond-gaussians}.
\begin{prop}
\label{prop:posterior-gaussian-redundancy}
Suppose Eq.~\eqref{eq:posterior-gaussian-assumption} holds with 
$\Sigma_Y\succ0$. Set $C_i=\Sigma_Y^{-1/2}\Sigma_{Y|X_i}\Sigma_Y^{-1/2}$, and let
$I_\cap^{\mathrm G}$ be the optimal value from the right-hand side
of Eq.~\eqref{eq:bl2}. 
Then $I_\cap\le I_\cap^{\mathrm G}$. Also, if every $\Sigma_{Y|X_i}\succ0$, 
\begin{align*}
I_\cap\le I_\cap^{\mathrm G}-D(P_Y\Vert P_{\tilde Y})\,,
\end{align*} where $\tilde Y$ is a Gaussian with the same mean
and covariance as $Y$.
\end{prop}
\begin{proof}
The law of total covariance gives $\Sigma_{Y|X_i}\preceq \Sigma_Y$, so $C_i\preceq\eye$.
For $M\succeq\eye$, 
let $E=M^{1/2}C_i^{1/2}$. The matrices 
$E^\top E=C_i^{1/2}MC_i^{1/2}$ and $EE^\top=M^{1/2}C_iM^{1/2}$ have the
same nonzero eigenvalues, so $E^\top E\preceq\eye$ iff $EE^\top\preceq\eye$. Therefore,
\begin{equation}
C_i^{1/2}MC_i^{1/2}\preceq\eye
\quad\Longleftrightarrow\quad
M^{1/2}C_iM^{1/2}\preceq\eye
\quad\Longleftrightarrow\quad C_i\preceq M^{-1}.
\label{eq:app-red-inverse-constraint}
\end{equation}

\paragraph{Positive-definite posterior covariances.}
First suppose every $\Sigma_{Y|X_i}\succ0$. For a probability distribution $P$ on
$\mathbb R^{d_Y}$ with smooth density $p$, write its Fisher information
matrix as
\begin{equation}
\mathcal J(P):=\int\nabla\ln p(y)\nabla\ln p(y)^\top\,P(dy).
\label{eq:appdeffim}
\end{equation}
Let $Q$ be any common lower channel, not necessarily Gaussian. Coupling $Q$
to each $X_i$ through its garbling kernel shows that, for $P_Q$-almost
every $q$, the posterior is a mixture of Gaussian posteriors,
\begin{equation}
\PYQ{q}=\int\mathcal N(m_i(x_i),\Sigma_{Y|X_i})\,P_{X_i|Q=q}(dx_i).
\end{equation}
These Gaussian convolutions have smooth densities. Define
\begin{equation}
\bar M_Q:=\Sigma_Y^{1/2}\int\mathcal J(\PYQ{q})\,P_Q(dq)\,\Sigma_Y^{1/2}.
\label{eq:barMb}
\end{equation}
Since each
Gaussian component has Fisher information matrix $\Sigma_{Y|X_i}^{-1}$, and  the Fisher information matrix is convex under mixtures
\citep[Prop.~4]{eskenazis2024entropy}, we have 
$\mathcal J(\PYQ{q})\preceq \Sigma_{Y|X_i}^{-1}$ for every $i$, hence $\bar M_Q\preceq \Sigma_Y^{1/2}\Sigma_{Y|X_i}^{-1}\Sigma_Y^{1/2}=C_i^{-1}$.

Applying the same inequality to $P_Y=\int\PYQ{q}\,P_Q(dq)$, followed
by the Cram\'er--Rao inequality~\citep[Thm.~20]{dembo1991information}, gives
\begin{equation}
\int\mathcal J(\PYQ{q})\,P_Q(dq)
\succeq\mathcal J(P_Y)\succeq \Sigma_Y^{-1}.
\end{equation}
Plugging into Eq.~\eqref{eq:barMb} gives $\eye=\Sigma_Y^{1/2}\Sigma_Y^{-1}\Sigma_Y^{1/2}
\preceq\bar M_Q$. 
Together, $\eye\preceq\bar M_Q\preceq C_i^{-1}$, so $C_i^{1/2}\bar M_QC_i^{1/2}\preceq\eye$ and $\bar M_Q$ is feasible
for the matrix program defining $I_\cap^{\mathrm G}$.

It remains to bound the information of $Q$. Gaussian convolution and
finite covariance ensure that $h(P_Y)$ and $D(P_Y\Vert P_{\tilde Y})$ are finite, with
$D(P_Y\Vert P_{\tilde Y})=h(P_{\tilde Y})-h(P_Y)$.
Data processing gives $I(Y;Q)\le I(Y;X_i)<\infty$, so $\int h(\PYQ{q})\,P_Q(dq)$ is
finite as well. Dembo's entropy--Fisher-information inequality
\citep[Eq.~2.36]{dembo1990information} gives
\begin{equation}
h(\PYQ{q})\ge
\frac12\ln\!\left((2\pi e)^{d_Y}|\mathcal J(\PYQ{q})^{-1}|\right).
\end{equation}
Averaging this bound over $q$ gives
\begin{equation}
\begin{aligned}
I(Y;Q)+D(P_Y\Vert P_{\tilde Y})&=h(P_{\tilde Y})-\int h(\PYQ{q})\,P_Q(dq) \\
&\le\frac12\int\ln\big|\Sigma_Y^{1/2}\mathcal J(\PYQ{q})\Sigma_Y^{1/2}\big|\,P_Q(dq)\\
&\le\frac12\ln\Big|\Sigma_Y^{1/2}\Big(\int\mathcal J(\PYQ{q})\,P_Q(dq)\Big)\Sigma_Y^{1/2}\Big|
=\frac12\ln\big|\bar M_Q\big|\le I_\cap^{\mathrm G}.
\end{aligned}
\label{eq:posterior-gaussian-redundancy-proof}
\end{equation}
In the last line, we used Jensen's inequality and concavity of $\ln \vert \cdot \vert$. 

\paragraph{Singular posterior covariances.}
Now allow $\Sigma_{Y|X_i}\succeq0$, and reduce to the positive-definite
case by continuity. Whiten the target to $Z=\Sigma_Y^{-1/2}(Y-\mathbb EY)$
and let $Z_t=\sqrt{1-t}\,Z+\sqrt t\,G$, with independent
$G\sim\mathcal N(0,\eye)$ and $0<t<1$. The smoothed target $Z_t$ again
satisfies Eq.~\eqref{eq:posterior-gaussian-assumption}, now with
$C_i(t)=(1-t)C_i+t\eye\succ0$, and any common lower channel $Q$ of the
original sources remains a common lower channel for $Z_t$. Since
$C_i(t)\succeq C_i$, Eq.~\eqref{eq:app-red-inverse-constraint} shows that
the smoothed feasible set is contained in the original one, so the
positive-definite case gives $I(Z_t;Q)\le I_\cap^{\mathrm G}$. As
$t\downarrow0$, $(Z_t,Q)\to(Z,Q)$ in distribution, and lower
semicontinuity of mutual information gives
$I(Y;Q)=I(Z;Q)\le\liminf_{t\downarrow0}I(Z_t;Q)\le I_\cap^{\mathrm G}$.
Taking the supremum over $Q$ proves $I_\cap\le I_\cap^{\mathrm G}$.
\end{proof}

\noindent\textbf{Theorem~\ref{thm:gaussian-optimality-redundancy} (restated).}
\begin{itshape}
\thmRedundancyStatement{\label{eq:thmredopt}}
\end{itshape}
\begin{proof}
First suppose $\sum_i C_i$ is singular, and choose a unit vector $v$
in the kernel. Every source determines
$Q=v^\top\Sigma_Y^{-1/2}(Y-\mu_Y)$ exactly, since its conditional
variance is $v^\top C_iv=0$. Thus $Q$ is a Gaussian common lower channel with
$I(Y;Q)=+\infty$, attaining $I_\cap$. 
This one-dimensional channel $\srcCh{Q}$ belongs to $\CondGauss(1)$.
Padding $Q$ with $d_Q-1$ independent standard Gaussian variables gives an
optimal channel in $\CondGauss(d_Q)$.

Now suppose 
$\sum_i C_i\succ0$.
Since $Y$ is Gaussian, $D(P_Y\Vert P_{\tilde Y})=0$, and
Proposition~\ref{prop:posterior-gaussian-redundancy} gives
$I_\cap\le I_\cap^{\mathrm G}$, the value of Eq.~\eqref{eq:bl2}.
The set $\mathcal R=\{R:C_i\preceq R\preceq\eye\ \forall i\}$ is compact
and nonempty, and every $R\in\mathcal R$ satisfies
$R\succeq n^{-1}\sum_i C_i\succ0$. Thus $R=M^{-1}$ gives an exact
correspondence with the feasible set of Eq.~\eqref{eq:bl2} by
Eq.~\eqref{eq:app-red-inverse-constraint}. The continuous objective
$-\tfrac12\ln|R|$ attains its maximum $I_\cap^{\mathrm G}$ at some
$\Rcap\in\mathcal R$, and by
Theorem~\ref{thm:posterior-gaussian-blackwell} the canonical
channel~\eqref{eq:canonical-posterior-channel} with $R=\Rcap$ is a common
lower channel attaining the matrix-program upper bound $I_\cap^{\mathrm G}$.
Thus, this channel is optimal and $I_\cap=I_\cap^{\mathrm G}$.

In this finite case, it remains to show that the channel can be realized
in $d_Q$ dimensions. For this, it suffices to show that
$\operatorname{rank}(\eye-\Rcap)\le d_Q$ (Sec.~\ref{app:gaussian-comparison}).
Since $C_i\preceq\Rcap\preceq\eye$, we have
$0\preceq\eye-\Rcap\preceq\eye-C_i$ for all $i$, so
$\operatorname{range}(\eye-\Rcap)\subseteq
\bigcap_i\operatorname{range}(\eye-C_i)$, whose dimension is $d_Q$.

To upper bound $d_Q$, we use the definition of $C_i$ and the Gaussian
conditioning formula to write
\begin{equation}
\eye-C_i=\Sigma_Y^{-1/2}(\Sigma_Y-\Sigma_{Y|X_i})\Sigma_Y^{-1/2}
=\Sigma_Y^{-1/2}\Sigma_{YX_i}\Sigma_{X_i}^{+}\Sigma_{X_iY}\Sigma_Y^{-1/2},
\label{eq:app-rank-bound}
\end{equation}
where $\Sigma_{X_i}^{+}$ is the pseudoinverse. Therefore,
$\operatorname{rank}(\eye-C_i)\le d_i$ and $d_Q\le\min(d_Y,\min_i d_i)$.
\end{proof}

\paragraph{Non-Gaussian sources.}
Except for the dimension bound, the proof above uses only that $Y$ is
Gaussian and that Eq.~\eqref{eq:posterior-gaussian-assumption} holds.
Hence $I_\cap=I_\cap^{\mathrm G}$ for Gaussian targets even when the
sources are not jointly Gaussian with the target, as stated in
Sec.~\ref{sec:beyond-gaussians}. 


\subsection{Gaussian optimality of union information}
\label{app:union-optimality}

We first establish the result for constant Gaussian posteriors, without
assuming that the target is Gaussian.
\begin{prop}
\label{prop:posterior-gaussian-union}
Suppose Eq.~\eqref{eq:posterior-gaussian-assumption} holds with every
$\Sigma_{Y|X_i}\succ0$ and $\Sigma_Y\succ0$. Set
$C_i=\Sigma_Y^{-1/2}\Sigma_{Y|X_i}\Sigma_Y^{-1/2}$ and let $I_\cup^{\mathrm G}$ and $\Wcup$ be the value and optimizer of
Eq.~\eqref{eq:union-convex}. Then $D(P_Y\Vert P_{\tilde Y})<\infty$ and
\begin{equation}
I_\cup=I_\cup^{\mathrm G}-D(P_Y\Vert P_{\tilde Y}),
\end{equation}
where $\tilde Y$ is a Gaussian with the same mean and covariance as $Y$.
The optimum is attained by a channel with constant Gaussian posterior covariance
$\Sigma_Y^{1/2}\Wcup \Sigma_Y^{1/2}$.
\end{prop}
\begin{proof}
Gaussian convolution and finite covariance ensure that $h(P_Y)$ and
$D(P_Y\Vert P_{\tilde Y})$ are finite, with $D(P_Y\Vert P_{\tilde Y})=h(P_{\tilde Y})-h(P_Y)$.
For any common upper channel $Q$ with $I(Y;Q)<\infty$,
Lemma~\ref{lem:garbling-covariance} gives
$\Sigma_{Y|Q}\preceq \Sigma_{Y|X_i}$
for all $i$, so $W_Q:=\Sigma_Y^{-1/2}\Sigma_{Y|Q}\Sigma_Y^{-1/2}\preceq C_i$.
Since $I(Y;Q)<\infty$, almost every posterior $\PYQ{q}$ is absolutely
continuous with respect to $P_Y$, which has a density. Each posterior
therefore has a density, so its entropy is defined and its covariance
$\Sigma_{Y|Q=q}$ is positive definite. In particular, $W_Q\succ0$ is
feasible for the union program. The Gaussian
maximum-entropy bound,
$h(\PYQ{q})\le\frac12\ln\!\left((2\pi e)^{d_Y}|\Sigma_{Y|Q=q}|\right)$,
then gives
\begin{equation}
\begin{aligned}
I(Y;Q)+D(P_Y\Vert P_{\tilde Y})&=h(P_{\tilde Y})-\int h(\PYQ{q})\,P_Q(dq) \\
&\ge-\frac12\int\ln\big|\Sigma_Y^{-1/2}\Sigma_{Y|Q=q}\Sigma_Y^{-1/2}\big|\,P_Q(dq)\\
&\ge-\frac12\ln\Big|\Sigma_Y^{-1/2}\Big(\int\Sigma_{Y|Q=q}\,P_Q(dq)\Big)\Sigma_Y^{-1/2}\Big|
=-\frac12\ln|W_Q|\ge I_\cup^{\mathrm G}.
\end{aligned}
\label{eq:posterior-gaussian-union-bound}
\end{equation}
In the last line, we used Jensen's inequality and concavity of $\ln|\cdot|$.
Channels with infinite information already satisfy this lower bound.

The matrix program has a positive-definite optimizer $\Wcup$: its
feasible set is compact and contains $\epsilon\eye$ for sufficiently
small $\epsilon>0$. Set $\Sigma^*_\cup=\Sigma_Y^{1/2}\Wcup \Sigma_Y^{1/2}$, draw
$X_1\sim P_{X_1}$, and construct
\begin{equation}
Q\sim\mathcal N\big(m_1(X_1),\Sigma_{Y|X_1}-\Sigma^*_\cup\big),\qquad
Y'=Q+Z,\quad Z\sim\mathcal N(0,\Sigma^*_\cup),
\label{eq:posterior-gaussian-union-construction}
\end{equation}
with $Z$ independent of $(X_1,Q)$. Given $X_1$, $Y'\sim\mathcal N(m_1(X_1),\Sigma_{Y|X_1})$, so by
Eq.~\eqref{eq:posterior-gaussian-assumption} the pair $(X_1,Y')$ has the
same distribution as $(X_1,Y)$; in particular, $Y'\sim P_Y$. Moreover,
$P_{Y'|Q=q}=\mathcal N(q,\Sigma^*_\cup)$.
Since $\Sigma^*_\cup\preceq \Sigma_{Y|X_i}$ for every $i$,
Theorem~\ref{thm:posterior-gaussian-blackwell} makes $Q$ a common upper
channel. Its information is
$h(P_Y)-\frac12\ln((2\pi e)^{d_Y}|\Sigma^*_\cup|)=I_\cup^{\mathrm G}-D(P_Y\Vert P_{\tilde Y})$,
attaining the bound.
\end{proof}

\noindent\textbf{Theorem~\ref{thm:gaussian-optimality-union} (restated).}
\begin{itshape}
\thmUnionStatement{\label{eq:thmunopt}}
\end{itshape}

\begin{proof}
Proposition~\ref{prop:posterior-gaussian-union} with $D(P_Y\Vert P_{\tilde Y})=0$ gives
$I_\cup=I_\cup^{\mathrm G}$, the value of Eq.~\eqref{eq:union-convex}. Its
optimizer $\Wcup$ is unique by strict convexity of $-\ln\det$, and by
Theorem~\ref{thm:posterior-gaussian-blackwell} the canonical channel with
$R=\Wcup$ is a common upper channel with information $I_\cup^{\mathrm G}$. Thus, this channel attains $I_\cup$, the infimum over all channels.

It
remains to show that it can be realized in $d_Q$ dimensions.
For this, it suffices to show that
$\operatorname{rank}(\eye-\Wcup)\le d_Q$ (Sec.~\ref{app:gaussian-comparison}). We
show that, in suitable orthonormal coordinates, $\Wcup=\operatorname{diag}(A^*,\eye)$
with $A^*\in\psd{d_Q}$, which implies this.

Let $\mathcal U:=\operatorname{range}\big(\sum_i(\eye-C_i)\big)$, so that $\dim\mathcal U=d_Q$. Also, since each $\eye-C_i\succeq0$,
$\mathcal U^\perp=\bigcap_i\ker(\eye-C_i)$. Intuitively, $\mathcal U$
contains the normalized target directions about which at least one
source is informative.  Let $O=[O_1\ O_2]$ be an orthogonal matrix in $\mathbb R^{d_Y\times d_Y}$ whose
$d_Q$ columns $O_1$ span $\mathcal U$ and whose remaining columns $O_2$ span
$\mathcal U^\perp$. Since $(\eye-C_i)O_2=0$, we have $C_iO_2=O_2$, and
therefore 
\begin{equation}
O^\top C_iO=
\begin{pmatrix}O_1^\top C_iO_1 & O_1^\top C_iO_2\\ O_2^\top C_iO_1 & O_2^\top C_iO_2\end{pmatrix}
=\begin{pmatrix}C_i' & 0\\ 0 & \eye\end{pmatrix},
\qquad C_i':=O_1^\top C_iO_1 .
\end{equation}
The map $W\mapsto O^\top WO$ preserves determinants and the Loewner order,
so $W$ is feasible for Eq.~\eqref{eq:union-convex} iff $O^\top WO\preceq
O^\top C_iO$ for all $i$, with the same objective value. We may therefore
work in these coordinates and, with slight abuse of notation, write $W$ and
$C_i=\operatorname{diag}(C_i',\eye)$ for the rotated matrices.

Write a feasible $W$ as
$W=\left(\begin{smallmatrix}A&B\\B^\top&D\end{smallmatrix}\right)$. The
diagonal blocks of $W\preceq C_i$ give $A\preceq C_i'$ and $D\preceq\eye$,
so $W'=\operatorname{diag}(A,\eye)$, which
discards all information about $\mathcal U^\perp$, is also feasible.
Moreover,
\begin{equation}
|W|\le|A|\,|D|\le|A|=|W'|,
\end{equation}
where the first inequality is Fischer's inequality for positive-definite
block matrices and the second follows from $D\preceq\eye$. Thus, by uniqueness,
$\Wcup=\operatorname{diag}(A^*,\eye)$. 

Finally, by Eq.~\eqref{eq:app-rank-bound},
$\operatorname{rank}(\eye-C_i)\le d_i$, so $d_Q\le\min(d_Y,\sum_i d_i)$.
\end{proof}

\subsection{Two-source closed form}
\label{app:two-source}

\noindent\textbf{Theorem~\ref{thm:two-source-union-closed-form} (restated).}
\begin{itshape}
\thmTwoSourceStatement{}
\end{itshape}

\begin{proof}
Set $\widehat W=V^\top WV$. The normalization $V^\top C_1V=\eye$ gives
$|V|^2|C_1|=1$, so $|W|=|\widehat W|/|V|^2=|C_1|\,|\widehat W|$. 
Since $V$
is invertible, congruence by $V$ preserves the Loewner order:
$W\preceq C_k$ iff $\widehat W\preceq V^\top C_kV$. The constraints
$0\preceq W\preceq C_1$ and $W\preceq C_2$ of
Eq.~\eqref{eq:union-convex} therefore become
\begin{equation}
0\leL\widehat W\leL V^\top C_1 V =\eye,
\qquad
\widehat W\leL V^\top C_2 V =\Lambda=\operatorname{diag}(\lambda_1,\dots,\lambda_{d_Y}).
\end{equation}
Thus every feasible \(\widehat W\) has diagonal entries
\(\widehat W_{jj}\le \min(1,\lambda_j)\). By Hadamard's inequality,
\begin{equation}
|\widehat W|\le\prod_{j=1}^{d_Y}\widehat W_{jj}
\le\prod_{j=1}^{d_Y}\min(1,\lambda_j).
\end{equation}
Since $|W|=|C_1|\,|\widehat W|$, this gives a lower bound on the union
objective in Eq.~\eqref{eq:union-convex}: every feasible $W$ satisfies
$-\frac12\ln|W|\ge-\frac12\ln|C_1|-\frac12\sum_j\ln\min(1,\lambda_j)$.
Equality is achieved by
\(\widehat W=\operatorname{diag}(\min(1,\lambda_j))\), and
\(W=V^{-\top}\widehat WV^{-1}\) gives the stated optimizer \(\Wcup\).
This gives
\begin{equation}
I_\cup =-\frac12\ln|\Wcup|
=-\frac12\ln|C_1|
-\frac12\sum_{j=1}^{d_Y}\ln\min(1,\lambda_j).
\end{equation}
Using $\ln\min(1,\lambda)=\frac12(\ln\lambda-|\ln\lambda|)$,
$\sum_j\ln\lambda_j=\ln|C_2|-\ln|C_1|$, and
$I(Y;X_i)=-\frac12\ln|C_i|$ gives the stated expression.
\end{proof}

\subsection{On the inclusion--exclusion principle for two sources}
\label{app:two-source-iep}
We now discuss the inclusion--exclusion principle for Blackwell redundancy
and union information with two sources. For sets,
$|A\cap B|=|A|+|B|-|A\cup B|$, suggesting the PID analogue 
\begin{equation}
I_\cap\stackrel{?}{=}I(Y;X_1)+I(Y;X_2)-I_\cup.
\end{equation}
Blackwell redundancy and union information, however, are defined
separately in Eq.~\eqref{eq:red}, and they need not satisfy
inclusion--exclusion~\citep{kolchinskyNovelApproachPartial2022}. We now
characterize exactly when they do for Gaussian sources, in three steps. First,
we drop the no-information constraint $\eye\preceq M$ from
Eq.~\eqref{eq:bl2} and solve the resulting relaxed problem in closed form,
using the same generalized eigenvectors as in
Theorem~\ref{thm:two-source-union-closed-form}. Second, we compare its
solution with the union information solution $\Wcup$, which shows that the relaxed problem satisfies
inclusion--exclusion exactly. Third, we restore the no-information
constraint.

Writing $R=M^{-1}$ and using Eq.~\eqref{eq:app-red-inverse-constraint},
the relaxed Blackwell redundancy problem is
\begin{equation}
I_\cap^{\mathrm{rel}}=\max_{R\succ0}\,-\frac12\ln|R|
\quad\where\quad R\succeq C_1,\ R\succeq C_2.
\label{eq:app-relaxed-and-union}
\end{equation}
With $V$ and $\Lambda$ as in Theorem~\ref{thm:two-source-union-closed-form},
set $\widehat R=V^\top RV$. Since $V^\top C_1V=\eye$ and
$V^\top C_2V=\Lambda$, we have $|R|=|C_1|\,|\widehat R|$, and the
constraints become $\widehat R\succeq\eye$ and $\widehat R\succeq\Lambda$.
These imply $\widehat R^{-1}\preceq\eye$ and
$\widehat R^{-1}\preceq\Lambda^{-1}$, so
$(\widehat R^{-1})_{jj}\le1/\max(1,\lambda_j)$. Hadamard's inequality
applied to $\widehat R^{-1}$ then gives
$|\widehat R|\ge\prod_j\max(1,\lambda_j)$, with equality only for
$\widehat R=\operatorname{diag}(\max(1,\lambda_j))$. The unique relaxed
optimizer is therefore
\begin{equation}
R_+=V^{-\top}\operatorname{diag}\big(\max(1,\lambda_j)\big)V^{-1}.
\label{eq:two-source-relaxed-meet}
\end{equation}
Compare this with the union optimizer
$\Wcup=V^{-\top}\operatorname{diag}(\min(1,\lambda_j))V^{-1}$ of
Theorem~\ref{thm:two-source-union-closed-form}: in each direction $j$, the
union keeps the more informative of the two sources, and the relaxed
redundancy keeps the less informative one. Since
$\min(1,\lambda_j)+\max(1,\lambda_j)=1+\lambda_j$ and
$\min(1,\lambda_j)\max(1,\lambda_j)=\lambda_j$,
\begin{equation}
R_++\Wcup=C_1+C_2,
\qquad
|R_+|\,|\Wcup|=|C_1|\,|C_2|.
\label{eq:app-matrix-iep}
\end{equation}

If $R_+\preceq\eye$, then $R_+$ also satisfies the no-information
constraint, so $\Mcap=R_+^{-1}$ solves Eq.~\eqref{eq:bl2}. The first
identity in Eq.~\eqref{eq:app-matrix-iep} then gives
Eq.~\eqref{eq:matrix-iep}, which justifies the redundancy fast path.

The second identity implies that 
$I_\cap^{\mathrm{rel}}=-\frac12\ln|R_+|$ satisfies inclusion--exclusion:
\[
I_\cap^{\mathrm{rel}}=I(Y;X_1)+I(Y;X_2)-I_\cup.
\]

Note that for two sources, the BROJA redundancy~\citep{bertschinger2014quantifying} is defined by
inclusion--exclusion, $I_\cap^{\mathrm{BR}}:=I(Y;X_1)+I(Y;X_2)-I_\cup^{\mathrm{BR}}$.
Since $I_\cup=I_\cup^{\mathrm{BR}}$, this gives
$I(Y;X_1)+I(Y;X_2)-I_\cup=I_\cap^{\mathrm{BR}}$, so $I_\cap^{\mathrm{rel}}=I_\cap^{\mathrm{BR}}$. Because the feasible set of
Eq.~\eqref{eq:bl2} is contained in that of the relaxed problem, we always
have
\begin{equation}
I_\cap\le 
I_\cap^{\mathrm{BR}}.
\end{equation}
So, Blackwell redundancy never exceeds BROJA redundancy. Since $R_+$ is
the unique relaxed optimizer, equality holds exactly when
$R_+\preceq\eye$.

\subsection{Calculations for the redundancy comparison}
\label{app:why-blackwell}

We now define the redundancy measures compared in
Sec.~\ref{sec:why-blackwell-redundancy}. Besides Blackwell redundancy, we
consider MMI, $I_\cap^{\min}$, $I_\cap^{\mathrm{PM}}$, $I_\cap^{\mathrm{ccs}}$, and
$I_\cap^{\mathrm{sx}}$, all of which are defined for any number of sources and
can be evaluated for continuous variables. We also
discuss the G\'acs--K\"orner redundancy, which we excluded from the
comparison. 
Note that Blackwell redundancy, MMI, $I_\cap^{\min}$, and $I_\cap^{\mathrm{ccs}}$ are
invariant under invertible rescaling of individual sources, but
$I_\cap^{\mathrm{PM}}$ and $I_\cap^{\mathrm{sx}}$ are not.

For each measure, we give its definition and evaluate it on the
COPY and COPY+ examples of Table~\ref{tab:why-blackwell}. Recall that $Y=(Y_1,Y_2)$, $X_1=Y_1+\xi_1$,
$X_2=Y_2+\xi_2$, and, for COPY+, $X_3=(Y_1+Y_2)/\sqrt2+\xi_3$, where
$Y_1,Y_2,\xi_1,\xi_2,\xi_3$ are independent standard Gaussians. We write
$X_i=u_i^\top Y+\xi_i$, with unit vectors $u_1=(1,0)$, $u_2=(0,1)$, and
$u_3=(1,1)/\sqrt2$. Thus every $X_i\sim\mathcal N(0,2)$, every
$P_{X_i|Y}$ is Gaussian with unit variance, and all sources are equally
informative, $I(Y;X_i)=\tfrac12\ln2$. Below, $p$ denotes the corresponding
Gaussian densities. For COPY,
every measure can be evaluated in closed form. For COPY+, MMI and
$I_\cap^{\min}$ are also closed form, and we evaluate the remaining
measures by Monte Carlo integration, averaging over $10^8$ samples of
$(Y,\xi_1,\xi_2,\xi_3)$ with the exact densities (standard errors below
$10^{-4}$ nats).

\paragraph{Blackwell redundancy.}
We use the matrix program of Eq.~\eqref{eq:bl2}. Here $\Sigma_Y=\eye$, so
$C_i=\Sigma_{Y|X_i}$, and observing $X_1$ halves the posterior variance of
$Y_1$ while leaving $Y_2$ untouched: $C_1=\operatorname{diag}(\frac12,1)$, and
likewise $C_2=\operatorname{diag}(1,\frac12)$. For COPY, the constraints of
Eq.~\eqref{eq:bl2} are therefore
$\eye\preceq M\preceq\operatorname{diag}(2,1)$ and
$\eye\preceq M\preceq\operatorname{diag}(1,2)$ (using
Eq.~\eqref{eq:app-red-inverse-constraint}). The first pair forces
$M_{22}=1$ and the second forces $M_{11}=1$. Then $M-\eye\succeq0$ has a
zero diagonal, so $M=\eye$, and $I_\cap=\frac12\ln|\eye|=0$.
Geometrically, the only ellipse that contains the unit disk and lies inside
both source ellipses is the unit disk itself. For COPY+, the third source
only adds a constraint, and $M=\eye$ remains feasible since $C_3\preceq\eye$,
so again $I_\cap=0$.

\paragraph{Minimum mutual information (MMI).}
MMI defines redundancy as the smallest source--target mutual
information~\citep{barrett_exploration_2015},
$I_\cap^{\mathrm{MMI}}=\min_i I(Y;X_i)$. Since every source has $I(Y;X_i)=\tfrac12\ln2$, MMI gives
$\tfrac12\ln2\approx0.35$ for both COPY and COPY+.

\paragraph{Minimum specific information, $I_\cap^{\min}$.}
\citet{williams2010nonnegative} define redundancy as the expected minimum
specific information. In continuous form~\citep{milzman2022signed},
\begin{equation}
 I_\cap^{\min}=\mathbb E_Y\min_i D\bigl(P_{X_i|Y}\Vert P_{X_i}\bigr).
\end{equation}
Since $P_{X_i}=\mathcal N(0,2)$ and $P_{X_i|Y=y}=\mathcal N(u_i^\top y,1)$,
every specific information has the same form,
$D(P_{X_i|Y=y}\Vert P_{X_i})=\frac12\ln2-\frac14+\frac14(u_i^\top y)^2$, so
\begin{equation}
I_\cap^{\min}=\tfrac12\ln2-\tfrac14+\tfrac14\,\mathbb E\min_i(u_i^\top Y)^2.
\end{equation}
The remaining expectation can be computed in closed form. Write $Y$ in
polar coordinates, $Y=r(\cos\theta,\sin\theta)$, where $r$ and $\theta$ are
independent, $\mathbb Er^2=2$, and $\theta$ is uniform on $[0,2\pi)$. With
$u_i=(\cos\phi_i,\sin\phi_i)$, we have $(u_i^\top Y)^2=r^2\cos^2(\theta-\phi_i)$. 
Since the integrand has period $\pi$, we have
\begin{equation}
\mathbb E\min_i(u_i^\top Y)^2=\frac2\pi\int_0^\pi\min_i\cos^2(\theta-\phi_i)\,d\theta .
\end{equation}
Evaluating this elementary integral gives $1-\frac2\pi$ for COPY
($\phi_i=0,\frac\pi2$), hence  $I_\cap^{\min}=\tfrac12\ln2-\tfrac{1}{2\pi}\approx0.19$. It gives $1-\frac{1+\sqrt2}{\pi}$ for COPY+
($\phi_i=0,\frac\pi4,\frac\pi2$), hence $I_\cap^{\min}=\tfrac12\ln2-\tfrac{1+\sqrt2}{4\pi}\approx0.15$.

\paragraph{Pointwise minimum, $I_\cap^{\mathrm{PM}}$.}
The pointwise redundancy of \citet{finn_pointwise_2018}, in the continuous
form of \citet{milzman2022signed}, is
\begin{equation}
 I_\cap^{\mathrm{PM}}=\mathbb E\min_i[-\ln p(X_i)]
                     -\mathbb E\min_i[-\ln p(X_i|Y)].
\end{equation}
For COPY, $-\ln p(X_i)=\frac12\ln(4\pi)+\frac12(X_i/\sqrt2)^2$
and $-\ln p(X_i|Y)=\frac12\ln(2\pi)+\frac12\xi_i^2$. Since $(X_1,X_2)/\sqrt2$ and
$(\xi_1,\xi_2)$ are both pairs of independent standard Gaussians, the two
expected minima differ only by the constants, and
$I_\cap^{\mathrm{PM}}=\tfrac12\ln2$. For COPY+, the standardized sources are
correlated, so this cancellation fails, and we evaluate the expectations
numerically.

\paragraph{Common change in surprisal, $I_\cap^{\mathrm{ccs}}$.}
\citet{ince2017measuring} evaluates local information quantities under
$P^{\mathrm{ME}}$, the maximum-entropy distribution with the same
source--target marginals $P_{Y,X_i}$ and the same source marginal
$P_{X_1,\ldots,X_n}$. For
Gaussians, these constraints fix the full covariance, so
$P^{\mathrm{ME}}=P$.
For $\varnothing\ne B\subseteq[n]:=\{1,\ldots,n\}$, let
$i_B=\ln[p^{\mathrm{ME}}(X_B|Y)/p^{\mathrm{ME}}(X_B)]$ be the local mutual information
between $Y$ and $X_B=(X_i)_{i\in B}$, where $p^{\mathrm{ME}}$ is the density. The local co-information is
$c=\sum_{B}(-1)^{|B|+1}i_B$, and
\begin{equation}
I_\cap^{\mathrm{ccs}}=\mathbb E_{P^{\mathrm{ME}}}\big[c\,\mathbf 1\{\operatorname{sign}c
=\operatorname{sign}i_{\{i\}}\ \forall i,\ \operatorname{sign}c=\operatorname{sign}i_{[n]}\}\big],
\end{equation}
that is, local co-information is counted only when its sign agrees with
that of every single-source and the full-joint local mutual information.
Following Ince's implementation, intermediate subsets are not
sign-tested. For COPY, the sources are independent and
conditionally independent given $Y$, so $i_{\{1,2\}}=i_{\{1\}}+i_{\{2\}}$
and $c=0$ pointwise, giving $I_\cap^{\mathrm{ccs}}=0$. For COPY+, we
evaluate the expectation numerically.

\paragraph{Shared exclusions, $I_\cap^{\mathrm{sx}}$.}
The shared-exclusion redundancy, introduced for discrete variables
by \citet{makkeh2021introducing} and extended to continuous variables by
\citet{ehrlich2024shared}, measures the information in the event that at least one source takes its observed
value. For continuous sources, the probability of this event is replaced by
its small-neighborhood limit, a sum of densities, so that
\begin{equation}
 I_\cap^{\mathrm{sx}}=\mathbb E\ln
 \frac{\sum_i p(X_i|Y)}{\sum_i p(X_i)}.
\end{equation}
For COPY,
$p(X_i)=\phi(X_i/\sqrt2)/\sqrt2$ and $p(X_i|Y)=\phi(\xi_i)$, where $\phi$ is the
standard normal density. Since $(X_1,X_2)/\sqrt2$ and $(\xi_1,\xi_2)$ have the
same distribution, $I_\cap^{\mathrm{sx}}=\tfrac12\ln2$. For COPY+, we
evaluate the expectation numerically.

\paragraph{G\'acs--K\"orner redundancy.}
\citet{griffith2014intersection} define redundancy as the information
about $Y$ in the common random variable of the
sources~\citep{gacs_common_1973}, that is, the most informative variable
that is a deterministic function of every source:
\begin{equation}
I_\cap^{\mathrm{GK}}=\sup\big\{I(Y;Q):\ Q=f_i(X_i)\text{ almost surely for all }i\big\}.
\end{equation}
This measure is trivial whenever the sources have a nondegenerate joint
Gaussian distribution, or more generally a joint density that is positive
everywhere. Indeed, suppose $f_1(X_1)=f_2(X_2)$ almost surely. Given
$X_1=x_1$, the source $X_2$ still has a positive density on all of
$\mathbb R^{d_2}$, and $f_2(X_2)=f_1(x_1)$ must hold almost surely. Hence
$f_2$ is almost everywhere equal to a constant, so the common variable is
trivial and $I_\cap^{\mathrm{GK}}=0$.
In particular, $I_\cap^{\mathrm{GK}}$ vanishes for COPY and COPY+, but also
for two noisy copies $X_i=Y+\xi_i$ of the same scalar target, where each
source's individual information is entirely redundant. This is why we excluded it
from the comparison in Sec.~\ref{sec:why-blackwell-redundancy}.



\section{Optimal control model}
\label{app:optimal-control}

This section describes the linear--quadratic--Gaussian (LQG) model of
Sec.~\ref{sec:optimal-control}.

We consider a particle with position $\mathbf r\in\mathbb R^2$, driven by a
linear restoring force with rotational drift $J \mathbf r$, applied control $\mathbf u$, and
standard Brownian noise:
\begin{equation}
 d\mathbf r=(J\mathbf r+\mathbf u)\,dt+d\mathbf W,
 \qquad J=\begin{pmatrix}-1&-0.5\\0.5&-0.3\end{pmatrix}.
 \label{ctrl:continuous}
\end{equation}
Without
control ($\mathbf u=0$), the particle fluctuates with stationary covariance
$\Sigma_0=\int_0^{\infty}e^{Js}e^{J^\top s}\,ds$, which solves the
Lyapunov equation $J\Sigma_0+\Sigma_0 J^{\top}+I=0$.
The controller samples the system at
intervals $\Delta t$ and holds each control fixed between samples.
Integrating Eq.~\eqref{ctrl:continuous} over one interval gives the exact
discrete-time dynamics
\begin{equation}
 \mathbf r_{t+1}=A\mathbf r_t+B\mathbf u_t+\mathbf w_t,
 \qquad A=e^{J\Delta t},
 \qquad B=\int_0^{\Delta t}e^{Js}\,ds=J^{-1}(A-I),
 \label{ctrl:discrete}
\end{equation}
where, from here on, $t=0,1,2,\ldots$ counts sampling steps (physical time
$t\Delta t$), and $\mathbf w_t$ is independent Gaussian process noise with
covariance
$\Sigma_{\mathbf w}=\int_0^{\Delta t}e^{Js}e^{J^\top s}\,ds
=\Sigma_0-A\Sigma_0 A^{\top}$.
The sensor reports $\mathbf s_t=\mathbf r_t+\mathbf n_t$, where $\mathbf n_t$
is independent Gaussian measurement noise with zero mean and covariance
\begin{equation}
 \Sigma_{\mathbf n}=\nu S,\qquad
 S=\begin{pmatrix}11&-4\\-4&2\end{pmatrix}.
 \label{ctrl:sensor}
\end{equation}
The fixed matrix $S$ sets the noise anisotropy and orientation; only its
strength $\nu$ varies.

The controller observes the measurements and applies
$\mathbf u_t=\mathbf u_t^{\mathrm{cmd}}+\boldsymbol\eta_t$. The
commanded control $\mathbf u_t^{\mathrm{cmd}}$ is chosen to minimize the 
long-run average of $\|\mathbf r_t\|^2+\|\mathbf u_t^{\mathrm{cmd}}\|^2$, while
$\boldsymbol\eta_t\sim\mathcal N(0,\sigma_\eta^2I)$ is independent actuator noise; we set $\sigma_\eta^2=0.05$. 
By the separation principle, the optimal controller
combines a Kalman filter with linear--quadratic regulator (LQR)
feedback~\citep{bertsekas2017dynamic}.

At each time $t$, the controller has two inputs: the current measurement
$\mathbf s_t$ and its memory $\mathbf m_t=\widehat{\mathbf r}_{t\mid t-1}$,
the Kalman prediction of $\mathbf r_t$ from the earlier measurements
$\mathbf s_0,\ldots,\mathbf s_{t-1}$~\citep{kalman1960new}.
In steady state, the controller performs three steps:
\begin{equation}
 \begin{aligned}
 \widehat{\mathbf r}_{t\mid t}
 &=L\mathbf s_t+(I-L)\mathbf m_t,\\
 \mathbf u_t
 &=-F\widehat{\mathbf r}_{t\mid t}+\boldsymbol\eta_t,\\
 \mathbf m_{t+1}&=(A-BF)\widehat{\mathbf r}_{t\mid t}.
 \end{aligned}
 \label{ctrl:controller}
\end{equation}
The first line is the Kalman update, which combines the prediction with the
new measurement using the steady-state Kalman gain $L$. The second applies
the LQR feedback $\mathbf u_t^{\mathrm{cmd}}=-F\widehat{\mathbf r}_{t\mid t}$
plus actuator noise. The third line propagates the estimate through the
known dynamics.

Both gains follow from discrete algebraic Riccati
equations~\citep{anderson1979optimal,bertsekas2017dynamic}. The Kalman gain
is $L=P(P+\Sigma_{\mathbf n})^{-1}$, where the prediction-error covariance
$P=\operatorname{Cov}(\mathbf r_t-\mathbf m_t)$ solves
\begin{equation}
 P=APA^{\top}-AP(P+\Sigma_{\mathbf n})^{-1}PA^{\top}+\Sigma_{\mathbf w}+\sigma_\eta^2BB^{\top},
\end{equation}
and the actuator noise contributes $\sigma_\eta^2BB^{\top}$ to the effective
process noise. The LQR gain is
$F=(I+B^{\top}\Pi B)^{-1}B^{\top}\Pi A$, where $\Pi$ solves
\begin{equation}
 \Pi=A^{\top}\Pi A-A^{\top}\Pi B(I+B^{\top}\Pi B)^{-1}B^{\top}\Pi A+I.
\end{equation}

Substituting Eq.~\eqref{ctrl:controller} into Eq.~\eqref{ctrl:discrete}, the
closed-loop state $\mathbf z_t=(\mathbf r_t,\mathbf m_t)$ evolves as
\begin{equation}
 \mathbf z_{t+1}=\mathcal A\mathbf z_t
 +\mathcal G_{\mathbf n}\mathbf n_t+\mathcal G_{\eta}\boldsymbol\eta_t+\mathcal G_{\mathbf w}\mathbf w_t,
 \qquad
 \mathcal A=\begin{pmatrix}A-BFL&-BF(I-L)\\ \bar AL&\bar A(I-L)\end{pmatrix},
\end{equation}
where $\bar A=A-BF$, $\mathcal G_{\mathbf n}=\big(\begin{smallmatrix}-BFL\\ \bar AL\end{smallmatrix}\big)$,
$\mathcal G_{\eta}=\big(\begin{smallmatrix}B\\0\end{smallmatrix}\big)$, and
$\mathcal G_{\mathbf w}=\big(\begin{smallmatrix}I\\0\end{smallmatrix}\big)$.
Since $\mathcal A$ is stable and the noises are independent, the stationary
covariance $\Sigma_{\mathbf z}$ of $\mathbf z_t$ solves the discrete Lyapunov
equation
\begin{equation}
 \Sigma_{\mathbf z}=\mathcal A\Sigma_{\mathbf z}\mathcal A^{\top}
 +\mathcal G_{\mathbf n}\Sigma_{\mathbf n}\mathcal G_{\mathbf n}^{\top}
 +\sigma_\eta^2\mathcal G_{\eta}\mathcal G_{\eta}^{\top}
 +\mathcal G_{\mathbf w}\Sigma_{\mathbf w}\mathcal G_{\mathbf w}^{\top}.
\end{equation}
Because $\mathbf s_t=\mathbf r_t+\mathbf n_t$ and
$\mathbf u_t=-FL\mathbf s_t-F(I-L)\mathbf m_t+\boldsymbol\eta_t$, with
$\mathbf n_t$ and $\boldsymbol\eta_t$ independent of $\mathbf z_t$, the joint
covariance of $(Y,X_1,X_2)=(\mathbf u_t,\mathbf s_t,\mathbf m_t)$ follows
directly from $\Sigma_{\mathbf z}$, $\Sigma_{\mathbf n}$, and $\sigma_\eta^2$.

Figure~\ref{ctrl:figure} fixes $\Delta t=0.15$ and sweeps the sensor
noise strength over $10^{-3}\leq\nu\leq10^{3}$ at $73$ logarithmically
spaced values.
The Kalman prediction, LQR gains, and redundancy/union/synergy values are recomputed at each point. 

\section{Newton--PCG optimization algorithm}
\label{sm:sec:optimization}

Here we describe our optimization algorithm for computing redundancy and union information. 
Both use the same method: Newton's method with matrix-free steps computed by preconditioned conjugate gradients
(PCG), constraints imposed by log-barriers, and primal--dual stopping bounds used to certify accuracy. After a common
preprocessing step, we describe the reduced problem, its bounds, and its
barrier formulas separately for redundancy (Sec.~\ref{sm:sec:optimization-redundancy})
and union information (Sec.~\ref{sm:sec:optimization-union}), followed by the
shared iteration scheme (Sec.~\ref{sm:sec:newton-pcg}). 

\subsection{Covariance preprocessing}
\label{sm:sec:optimization-preprocessing}
Redundancy and union information depend only on the source--target
marginals, and they are invariant under invertible affine transformations
of the target and of each source separately. Both solvers therefore take
as input only the covariance matrix of each pair $(X_i,Y)$. 
We retain the assumption $\Sigma_Y\succ0$, but allow singular source--target
marginals. If a source covariance is singular, we first discard its
zero-variance directions, which contain no information. With
$\Sigma_Y=L_YL_Y^\top$, the posterior covariances in whitened target
coordinates can be computed as 
\begin{equation}
 C_i=\eye-L_Y^{-1}\Sigma_{YX_i}\Sigma_{X_i}^{+}\Sigma_{X_iY}L_Y^{-\top}.
\end{equation}
Here $\Sigma_{X_i}^{+}$ denotes the Moore--Penrose pseudoinverse.



\paragraph{Finite-noise whitening.}
If needed, we transform the marginals to whitened coordinates
in which the target has identity covariance and each source is observed
with standard Gaussian noise:
\begin{equation}
 \tilde X_i=H_i\tilde Y+N_i,\qquad \tilde Y\sim\mathcal N(0,\eye),\qquad N_i\sim\mathcal N(0,\eye).
\end{equation}
The covariance blocks below refer to the marginals after discarding
zero-variance source directions and, for redundancy, any smoothing. Specifically, with Cholesky factorizations
\begin{equation}
 \Sigma_Y=L_YL_Y^\top,\qquad
 \Gamma_i=\Sigma_{X_iY}L_Y^{-\top},\qquad
 \Sigma_{X_i}-\Gamma_i\Gamma_i^\top=L_iL_i^\top,
 \label{sm:eq:preprocessing}
\end{equation}
we set $\tilde Y=L_Y^{-1}(Y-\mu_Y)$, $\tilde X_i=L_i^{-1}(X_i-\mu_{X_i})$, and
$H_i=L_i^{-1}\Gamma_i\in\mathbb R^{d_i\times d_Y}$. Here
$\Sigma_{X_i}-\Gamma_i\Gamma_i^\top=\Sigma_{X_i|Y}$ is the
conditional noise covariance, and the factor $L_Y$ is shared by all sources.
In these coordinates, the normalized posterior covariances of
Sec.~\ref{sec:gaussian-pid} are $C_i=(\eye+H_i^\top H_i)^{-1}\succ0$. 

\subsection{Redundancy}
\label{sm:sec:optimization-redundancy}
\paragraph{Singular $C_i$.}
Before proceeding, we consider the case where some $C_i$ may be singular. First, observe that redundancy is infinite if $\sum_i C_i$ is singular. Thus, we consider the case where $\sum_iC_i\succ0$ but some $C_i$ are singular.

We smooth the channels as in the proof of
Proposition~\ref{prop:posterior-gaussian-redundancy}: for $0<t<1$, replace the
whitened target $\tilde Y$ by $\tilde Y_t=\sqrt{1-t}\,\tilde Y+\sqrt t\,G$, with independent
$G\sim\mathcal N(0,\eye)$. This gives smoothed channels with
$C_i(t)=(1-t)C_i+t\eye\succ0$, whose redundancy we denote $I_\cap(t)$. Let
$\gamma=\lambda_{\min}(n^{-1}\sum_iC_i)$. Here $0<\gamma<1$: the lower bound
holds because $\sum_iC_i\succ0$, and the upper bound because some unit
vector $v$ has $v^\top C_jv=0$, while $v^\top C_iv\le1$ for all $i$. Then
we have
\begin{equation}
 0\le I_\cap-I_\cap(t)\le
 E(t):=\frac{d_Y}{2}\ln\bigl(1+t(\gamma^{-1}-1)\bigr).
 \label{sm:eq:singular-smoothing-error}
\end{equation}
Every feasible posterior $R$ of the original problem satisfies
$R\succeq n^{-1}\sum_iC_i\succeq\gamma\eye$, and
$R_t=(1-t)R+t\eye$ is feasible for the smoothed problem. Applying this to
the original optimizer and bounding $\ln|R_t|-\ln|R|$ eigenvalue by
eigenvalue gives Eq.~\eqref{sm:eq:singular-smoothing-error}; the lower bound
follows because smoothing shrinks the feasible set.
We choose the smoothing weight once, before optimization:
\begin{equation}
 t=\min\left\{\frac12,\;
 \frac{e^{\epsilon_{\rm abs}/d_Y}-1}{\gamma^{-1}-1}\right\},
 \label{sm:eq:singular-smoothing-weight}
\end{equation}
so $E(t)\le\epsilon_{\rm abs}/2$, and apply the solver described below to the
smoothed channels with tolerance $\epsilon_{\rm abs}^\prime =\epsilon_{\rm abs}/2$.
If the smoothed solver returns a certified interval $[I_t^{\mathrm{LB}},I_t^{\mathrm{UB}}]$, then
$[I_t^{\mathrm{LB}},I_t^{\mathrm{UB}}+E(t)]$ bounds the original redundancy and has width at most
$\epsilon_{\rm abs}$; we return its midpoint.
The smoothed channels have $C_i(t)\succ0$, so the whitening of
Sec.~\ref{sm:sec:optimization-preprocessing} applies to them.

\paragraph{Reduction.}
To compute redundancy, we first restrict Eq.~\eqref{eq:bl2} to the
subspace of target directions observed by every source. In the whitened
coordinates, the constraints of Eq.~\eqref{eq:bl2} read
$\eye\preceq M\preceq C_i^{-1}=\eye+A_i$, with $A_i:=H_i^\top H_i$. Every
feasible $M$ therefore satisfies $0\preceq M-\eye\preceq A_i$ for all $i$,
so $\operatorname{range}(M-\eye)$ lies in the common range
\begin{equation}
 \mathcal C:=\bigcap_i\operatorname{range}(A_i)=\bigcap_i\operatorname{range}(H_i^\top),
\end{equation}
the set of target directions observed by every source. Restricting to
$\mathcal C$ shrinks the problem from $d_Y$ to
$d_Q=\dim\mathcal C\le\min(d_Y,\min_i d_i)$ dimensions, where $d_Q$ is as in Theorem~\ref{thm:gaussian-optimality-redundancy}, and removes the
directions in which the constraints are degenerate.

If $d_Q=0$, redundancy is zero. Otherwise, every feasible $M$ has the form
$M=\eye+PKP^\top$ with $K\in\psd{d_Q}$, where the columns of
$P\in\mathbb R^{d_Y\times d_Q}$ form an orthonormal basis of $\mathcal C$. The constraint
$M\preceq\eye+A_i$ then reads $PKP^\top\preceq A_i$, which, because
$\operatorname{range}(P)\subseteq\operatorname{range}(A_i)$, is equivalent to
$K\preceq(P^\top A_i^\dagger P)^{-1}$. With $D_i:=(P^\top A_i^\dagger P)^{-1}$, the problem becomes
\begin{equation}
 f(K)=-\ln|\eye+K|,\qquad
 I_\cap=-\frac12\min_{0\preceq K\preceq D_i\ \forall i}f(K).
 \label{sm:eq:redundancy-program}
\end{equation}

The orthonormal basis $P$ can be computed as follows. 
The SVD of each $H_i$
gives an orthonormal basis $B_i^\perp$ of $\ker H_i=\ker A_i$.
A direction lies in $\mathcal C$ iff it is orthogonal to
the columns of every $B_i^\perp$, so $P$ is an orthonormal basis of the null space of the
stacked matrix $[B_1^\perp\ \cdots\ B_n^\perp]^\top$, which can be obtained using one more
SVD. Working with the $H_i$
rather than the $A_i$ avoids squaring small singular values. 
If every
$H_i$ has full column rank, $P=\eye$.

\paragraph{Bounds.}
A feasible $K$ gives the primal upper bound $f(K)$ on the minimum. For a
lower bound, attach multipliers $Z_i\succeq0$ to the constraints
$K\preceq D_i$ and $Z_0\succeq0$ to $K\succeq0$, and let
$Z:=\sum_iZ_i-Z_0$. For feasible $U$, the added terms
$\operatorname{tr}[Z_i(U-D_i)]$ and $-\operatorname{tr}(Z_0U)$ are
nonpositive, so
\begin{equation}
 \begin{aligned}
 \min_{U\ \mathrm{feasible}}f(U)
 &\ge\inf_{\eye+U\succ0}\Big\{-\ln|\eye+U|+\sum_i\operatorname{tr}[Z_i(U-D_i)]-\operatorname{tr}(Z_0U)\Big\}\\
 &=\ln|Z|+d_Q-\operatorname{tr}Z-\sum_i\operatorname{tr}(Z_iD_i)=:\ell,
 \end{aligned}
 \label{sm:eq:optimization-dual-bound}
\end{equation}
where the inequality also enlarges the domain to $\eye+U\succ0$. When
$Z\succ0$, the expression in braces is convex in $U$ with gradient
$-(\eye+U)^{-1}+Z$, so the infimum is attained at $\eye+U=Z^{-1}$, which
gives the second line; otherwise, the infimum is $-\infty$.

Any multipliers give a valid bound, and we choose them from the current
iterate. At a strictly feasible $0\prec K\prec D_i$ and barrier weight
$\mu>0$, we take the barrier multipliers $Z_i=\mu(D_i-K)^{-1}$
and let $\bar Z:=\sum_iZ_i$. We then
choose $Z_0$ to make the bound as tight as possible: maximizing $\ell$ over
$Z_0\succeq0$ amounts to maximizing $\ln|Z|-\operatorname{tr}Z$ subject to
$0\prec Z\preceq\bar Z$. The maximizer is
\begin{equation}
 Z=\min(\eye,\bar Z),
 \label{sm:eq:optimal-lower-multiplier}
\end{equation}
obtained by replacing each eigenvalue $z_j$ of $\bar Z$ by $\min(1,z_j)$: it
is feasible, and $Z^{-1}-\eye\succeq0$ is a multiplier for the constraint
$Z\preceq\bar Z$ that satisfies complementary slackness. Writing
$\ell(K,\mu)$ for the resulting bound gives the information interval
\begin{equation}
 -\tfrac12f(K)\le I_\cap\le-\tfrac12\ell(K,\mu).
\end{equation}

\paragraph{Barrier.}
The barrier objective is
\begin{equation}
 \Phi_\mu(K)=-\ln|\eye+K|-\mu\ln|K|-\mu\sum_i\ln|D_i-K|.
 \label{sm:eq:optimization-barrier}
\end{equation}
Writing $\Theta=\eye+K$, the gradient, Hessian action on a
symmetric matrix $E$, and preconditioner factor are
\begin{equation}
 \begin{aligned}
 \nabla\Phi_\mu(K)&=-\Theta^{-1}-\mu K^{-1}+\mu\sum_i(D_i-K)^{-1},\\
 \nabla^2\Phi_\mu(K)[E]&=\Theta^{-1}E\Theta^{-1}+\mu K^{-1}EK^{-1}+\mu\sum_i(D_i-K)^{-1}E(D_i-K)^{-1},\\
 T_\mu(K)&=\Theta^{-1}+\sqrt\mu\Big(K^{-1}+\sum_i(D_i-K)^{-1}\Big).
 \end{aligned}
 \label{sm:eq:optimization-newton}
\end{equation}
Using $\operatorname{tr}[Z_i(D_i-K)]=d_Q\mu$, the primal--dual gap can be
written as a sum of nonnegative terms,
\begin{equation}
 f(K)-\ell(K,\mu)
 =\sum_j\big(\lambda_j-1-\ln\lambda_j\big)+\operatorname{tr}(Z_0K)+nd_Q\mu,
 \label{sm:eq:direct-gap}
\end{equation}
where $\lambda_j$ are the eigenvalues of $\Theta^{1/2}Z\Theta^{1/2}$. At an
exact barrier minimizer, the barrier multiplier $Z_0=\mu K^{-1}$ gives
$Z=\Theta^{-1}$, so every $\lambda_j=1$ and
\begin{equation*}
 f(K)-\ell(K,\mu)=(n+1)d_Q\mu.
\end{equation*}
The gap obtained with Eq.~\eqref{sm:eq:optimal-lower-multiplier} is at most
this large, so it shrinks in proportion to $\mu$ along the central path.

\subsection{Union information}
\label{sm:sec:optimization-union}

We assume without loss of generality that $C_i\succ 0$, since $I_\cup$ is infinite if any $C_i$ is singular.  

\paragraph{Reduction.}
To compute union information, we drop target directions that no source
observes. As in the proof of
Theorem~\ref{thm:gaussian-optimality-union} (Sec.~\ref{app:union-optimality}),
let $O=[P_\cup\ P_\cup^\perp]$ be an orthogonal matrix whose first $r$
columns $P_\cup\in\mathbb R^{d_Y\times r}$ span a subspace containing
$\mathcal U=\operatorname{range}\big(\sum_i(\eye-C_i)\big)$. Since each $\eye-C_i$ vanishes on $\operatorname{range}(P_\cup^\perp)\subseteq\mathcal U^\perp$,
\begin{equation}
 C_i=O\begin{pmatrix}P_\cup^\top C_iP_\cup&0\\0&\eye\end{pmatrix}O^\top,
\end{equation}
and the Fischer-inequality argument of that proof, which uses only this
block form, shows that the union optimizer has the same block form,
$\Wcup=O\operatorname{diag}(W_r,\eye)O^\top$. For $W=O\operatorname{diag}(W_r,\eye)O^\top$, we have
$\ln|W|=\ln|W_r|$, and $W\preceq C_i$ iff $W_r\preceq P_\cup^\top C_iP_\cup$.
It therefore suffices to solve Eq.~\eqref{eq:union-convex} for the
$r\times r$ matrices $P_\cup^\top C_iP_\cup$.

It remains to construct $P_\cup$. In whitened coordinates,
$\eye-C_i=H_i^\top(\eye+H_iH_i^\top)^{-1}H_i$ has the same range as
$H_i^\top$, so $\mathcal U$ is the column space of
$[H_1^\top\ \cdots\ H_n^\top]$. If $\sum_i d_i\ge d_Y$, we take $P_\cup=\eye$
and $r=d_Y$. Otherwise, the stacked matrix $[H_1^\top\ \cdots\ H_n^\top]$ has
$d_Y$ rows and $r:=\sum_id_i<d_Y$ columns, and we compute its thin QR
factorization
\begin{equation}
 [H_1^\top\ \cdots\ H_n^\top]=P_\cup T,
\end{equation}
where $P_\cup\in\mathbb R^{d_Y\times r}$ has orthonormal columns and
$T\in\mathbb R^{r\times r}$ is upper triangular. Then $H_iP_\cup$ is the $i$th
row block of $T^\top$, and
$P_\cup^\top C_iP_\cup=\big(\eye+(H_iP_\cup)^\top H_iP_\cup\big)^{-1}$ follows
from $T$ alone, without forming $P_\cup$ or the full $C_i$. We keep all $r$
QR directions, even when the stacked channel matrix is rank deficient, so that no weakly observed direction
is discarded by a numerical rank decision.

Writing $K\in\psd{r}$ for the reduced matrix $W_r$ and
$D_i:=P_\cup^\top C_iP_\cup$ for the reduced bounds, the problem is
\begin{equation}
 f(K)=-\ln|K|,\qquad
 I_\cup=\frac12\min_{0\preceq K\preceq D_i\ \forall i}f(K).
 \label{sm:eq:union-program}
\end{equation}

\paragraph{Bounds.}
A feasible $K$ gives the primal upper bound $f(K)$ on the minimum. For a
lower bound, attach multipliers $Z_i\succeq0$ to the constraints
$K\preceq D_i$, and let $Z:=\sum_iZ_i$. As for redundancy,
\begin{equation}
 \min_{U\ \mathrm{feasible}}f(U)
 \ge\inf_{U\succ0}\Big\{-\ln|U|+\sum_i\operatorname{tr}[Z_i(U-D_i)]\Big\}
 =\ln|Z|+r-\sum_i\operatorname{tr}(Z_iD_i)=:\ell,
 \label{sm:eq:union-dual-bound}
\end{equation}
where, if $Z\succ0$, the infimum is attained at $U=Z^{-1}$; otherwise it is $-\infty$. At a strictly feasible
$0\prec K\prec D_i$ and barrier weight $\mu>0$, we take
$Z_i=\mu(D_i-K)^{-1}$, so that $Z\succ0$ automatically. Writing
$\ell(K,\mu)$ for the resulting bound gives the information interval
\begin{equation}
 \tfrac12\ell(K,\mu)\le I_\cup\le\tfrac12f(K).
\end{equation}
The bound also
applies to the full, unreduced problem: extending each reduced multiplier
$Z_i$ to $P_\cup Z_iP_\cup^\top+(\eye-P_\cup P_\cup^\top)/n$ makes their sum
the identity on the discarded directions, and their contribution to the
full bound cancels.

\paragraph{Barrier.}
The barrier objective is
\begin{equation}
 \Phi_\mu(K)=-\ln|K|-\mu\sum_i\ln|D_i-K|.
 \label{sm:eq:union-barrier}
\end{equation}
The gradient, Hessian action on a symmetric matrix $E$, and
preconditioner factor are
\begin{equation}
 \begin{aligned}
 \nabla\Phi_\mu(K)&=-K^{-1}+\mu\sum_i(D_i-K)^{-1},\\
 \nabla^2\Phi_\mu(K)[E]&=K^{-1}EK^{-1}+\mu\sum_i(D_i-K)^{-1}E(D_i-K)^{-1},\\
 T_\mu(K)&=K^{-1}+\sqrt\mu\sum_i(D_i-K)^{-1}.
 \end{aligned}
 \label{sm:eq:union-newton}
\end{equation}
Using $\operatorname{tr}[Z_i(D_i-K)]=r\mu$, the primal--dual gap can be
written as a sum of nonnegative terms,
\begin{equation}
 f(K)-\ell(K,\mu)=\sum_j\big(\lambda_j-1-\ln\lambda_j\big)+nr\mu,
 \label{sm:eq:union-direct-gap}
\end{equation}
where $\lambda_j$ are the eigenvalues of $K^{1/2}ZK^{1/2}$. At an exact
barrier minimizer, $Z=K^{-1}$, so every $\lambda_j=1$ and
\begin{equation*}
 f(K)-\ell(K,\mu)=nr\mu.
\end{equation*}


\subsection{Newton--PCG iterations}
\label{sm:sec:newton-pcg}
For two-source union information, the closed form of
Theorem~\ref{thm:two-source-union-closed-form} can replace the iterations
below (``Closed form'' in Fig.~\ref{fig:pid-value-timing}).
For two-source redundancy, the fast path of
Sec.~\ref{sec:numerical-algorithms} is used when the candidate $R$ satisfies
$R\preceq\eye$.
Otherwise, we use the iterations below. They are shared by both problems:
in both, $K\in\psd{m}$ is the reduced matrix and $D_i$ are its bounds, with
$m=d_Q$ and $f$, $\ell$, $\Phi_\mu$, $T_\mu$ as in
Sec.~\ref{sm:sec:optimization-redundancy} for redundancy, and $m=r$ with the
corresponding quantities of Sec.~\ref{sm:sec:optimization-union} for union.

For conditioning, we whiten the bounds by their mean,
$n^{-1}\sum_iD_i=L_DL_D^\top$, replacing $K$ and $D_i$ by
$L_D^{-1}K L_D^{-\top}$ and $L_D^{-1}D_iL_D^{-\top}$ (and, for redundancy,
the identity in $f(K)=-\ln|\eye+K|$ by $L_D^{-1}L_D^{-\top}$).
Due to this whitening, the objective and its bounds include the additive constant $-2\ln|L_D|$;
all formulas below use the original coordinates.

Initialize $\mu=\mu_0>0$, $f^{\mathrm{LB}}=-\infty$, and
$f^{\mathrm{UB}}=0$ for redundancy or $+\infty$ for union.
Take $K$ to be half the largest multiple of $(\sum_iD_i^{-1})^{-1}$
satisfying $K\preceq D_i$ for every $i$, giving $0\prec K\prec D_i$.
Repeat the following steps:
\begin{enumerate}
\item \textbf{Update the bounds.} Set
      \begin{equation}
       f^{\mathrm{UB}}\leftarrow\min\{f^{\mathrm{UB}},f(K)\},\qquad
       f^{\mathrm{LB}}\leftarrow\max\{f^{\mathrm{LB}},\ell(K,\mu)\}.
      \end{equation}
      If the information interval has width
      $\frac12(f^{\mathrm{UB}}-f^{\mathrm{LB}})\le\epsilon_{\rm abs}$, the
      requested absolute accuracy, return its midpoint,
      $\frac14(f^{\mathrm{UB}}+f^{\mathrm{LB}})$ for union or
      $-\frac14(f^{\mathrm{UB}}+f^{\mathrm{LB}})$ for redundancy.
\item \textbf{Compute a direction.} With $g=\nabla\Phi_\mu(K)$, use CG
      to approximately solve the Newton equation with preconditioner
      $\mathcal P$:
      \begin{equation}
       \nabla^2\Phi_\mu(K)[\Delta]=-g,\qquad
       \mathcal P(E)=T_\mu^{-1}(K)ET_\mu^{-1}(K).
      \end{equation}
      Replace a non-descent direction by $-\mathcal P(g)$.
\item \textbf{Check centering.} If
      $-\langle g,\Delta\rangle/2\le\mu^2$ (an estimate of half the squared Newton decrement),
      or if $N_{\max}$ Newton steps have already been taken at the current
      $\mu$, set $\mu\leftarrow\rho\mu$ and return to step~1 with the same
      $K$. Otherwise, proceed at the current barrier weight.
\item \textbf{Take a step.} Starting from $\alpha=1$, halve $\alpha$
      until $K+\alpha\Delta$ is strictly feasible (checked by Cholesky) and
      \begin{equation}
       \Phi_\mu(K+\alpha\Delta)
       \le\Phi_\mu(K)+c\alpha\langle g,\Delta\rangle.
      \end{equation}
      Set $K\leftarrow K+\alpha\Delta$ and return to step~1. If no
      step is accepted after $N_{\mathrm{ls}}$ halvings, keep $K$ and set
      $\mu\leftarrow\rho\mu$.
\end{enumerate}
The full optimization
loop (CG, line search, barrier continuation, and certification) is
just-in-time compiled in JAX. CG stops at a relative Frobenius residual of
$0.1$ or after $\max(100,2m)$ iterations. The remaining parameters are
$\mu_0=0.1$, $\rho=0.25$, $N_{\max}=30$, $N_{\mathrm{ls}}=80$, and the
line-search constant $c=10^{-4}$.

\subsection{Benchmark instances and comparison methods}
\label{sm:sec:experiments}

The random channels in Fig.~\ref{fig:pid-value-timing} are
$X_i=H_iY+N_i$, where $Y$ and the noises $N_i$
are independent standard Gaussians and $(H_i)_{jk}\sim\mathcal N(0,1/d_Y)$.
We generate five matched random instances per setting. 
The five sweeps vary all dimensions
together, one source dimension, both source dimensions together, the
target dimension, and the source count, as specified in the figure caption.
For our certified stopping rule, we use
$\epsilon_{\rm abs}=0.01$ nats.

All timings were measured on a MacBook Pro with an Apple M4 Pro chip
(macOS~15.7, Python~3.12, NumPy~2.5, and JAX~0.7.1), using a single CPU thread and 64-bit
floating point. 
Every matrix shape is compiled and warmed up before
measurement. The reported times exclude JAX compilation, input
generation, and independent final accuracy checks. Each call has a
$300$-second wall-clock limit. 
Ordinary markers
require five successful runs. Crosses indicate at least
one $300$-s timeout. 

\paragraph{Certified accuracy.}
Our formulation returns explicit bounds $I^{\mathrm{LB}}\le I\le I^{\mathrm{UB}}$ on the globally
optimal information value and reports success only when
$I^{\mathrm{UB}}-I^{\mathrm{LB}}\le\epsilon_{\rm abs}$. The original gradient-based BROJA
implementations~\citep{venkatesh2023gaussian,zhao2025partial} instead stop
when the objective stagnates, which does not certify optimality.

Thus, for each matched instance, we first obtain our certified union interval
$[I^{\mathrm{LB}},I^{\mathrm{UB}}]$. We then run each competing method until
its feasible objective reaches or improves on
$I^{\mathrm{UB}}$, or until the time limit is reached. The competitors do not compute lower bounds or certify
an optimality gap themselves. This protocol deliberately favors them, because 
our timing includes certification. 

\paragraph{Source-space BROJA union information (Newton).}
We compare our target-space union solver with the existing approach to
Gaussian BROJA union information, which optimizes over the joint
distribution of the sources rather than over a single $d_Y\times d_Y$
matrix. Because union information is the same under both definitions,
the two approaches compute the same quantity, so their runtimes can be
compared directly.
For whitened channels, stack $H=(H_1^\top,\ldots,H_n^\top)^\top$ and let
$S$ be the joint covariance of the source noises. Fixing each diagonal
block $S_{ii}=\eye$ preserves every source--target marginal, so
\begin{equation}
 I_\cup=\inf_{S\succ0:\,S_{ii}=\eye}\varphi(S),
 \qquad \varphi(S)=\tfrac12\ln|\eye+H^\top S^{-1}H|,
 \label{sm:eq:source-union-program}
\end{equation}
where the infimum includes limits at singular boundary points. For two
sources, this is the cross-noise formulation of \citet{venkatesh2023gaussian}
and \citet{zhao2025partial}.
``BROJA (Newton)'' is our implementation: starting from $S=\eye$, it
minimizes $2\varphi(S)-\mu\ln|S|$ over the full
$(\sum_i d_i)\times(\sum_i d_i)$ matrix $S$ along the log-barrier path,
using the same compiled Newton--PCG driver as our target-space solver
with a Hessian-diagonal preconditioner.


\paragraph{Source-space BROJA union information (Thin-PID).}
Fig.~\ref{sm:fig:pid-value-timing} adds ``BROJA (Thin-PID)'', the
RProp-based Thin-PID optimizer of \citet{zhao2025partial} for the same
source-space problem, with the same stopping rule.


\begin{figure}[t]
\centering
\includegraphics[width=\textwidth]{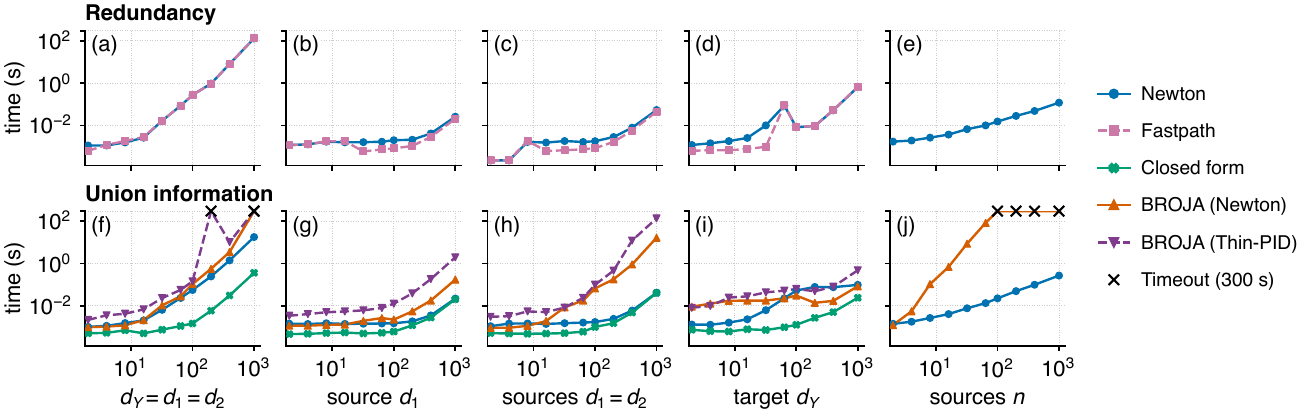}
\caption{\label{sm:fig:pid-value-timing}
Extended comparison on the same instances as Fig.~\ref{fig:pid-value-timing},
with the same panel definitions and BROJA (Thin-PID) added.
Our iterative methods certify a gap of at most $0.01$ nats; the BROJA
methods stop against our certified upper bound.
Ordinary markers average five successful runs. Crosses indicate at least
one $300$-s baseline timeout. 
}
\end{figure}

\end{document}